\documentclass{article}
\usepackage{verbatim}
\usepackage{mathtools}
\usepackage{amssymb}
\usepackage{mathabx}
\usepackage{bm}
\usepackage{stmaryrd}
\usepackage{hyperref}
\usepackage{graphicx}
\usepackage{amsmath}

\newtheorem{theorem}{Theorem}

\newtheorem{corollary}[theorem]{Corollary}

\newtheorem{definition}[theorem]{Definition}
\newtheorem{example}[theorem]{Example}

\newtheorem{lemma}[theorem]{Lemma}

\newtheorem{proposition}[theorem]{Proposition}
\newtheorem{remark}[theorem]{Remark}

\newenvironment{proof}[1][Proof]{\textbf{#1.} }{\ \rule{0.5em}{0.5em}}

\newcommand{\Sym}{\mathrm{Sym}}

\newcommand{\so}{\mathfrak{so}}

\newcommand{\h}{\mathfrak{h}}

\newcommand{\gl}{\mathfrak{gl}}

\newcommand{\T}{\mathbf{T}}

\newcommand{\SO}{\mathrm{SO}}
\newcommand{\OGr}{\mathrm{OGr}}
\newcommand{\GL}{\mathrm{GL}}

\renewcommand{\P}{\mathbf{P}}

\newcommand{\Spin}{\mathrm{Spin}}

\newcommand{\Pf}{\mathrm{Pf}}

\newcommand{\n}{\mathfrak{n}}
\newcommand{\ddef}{\mathrm{def}}

\newcommand{\rT}{\mathrm{T}}
\newcommand{\depth}{\mathrm{depth}}
\newcommand{\rht}{\mathrm{ht}}
\newcommand{\Map}{\mathrm{Map}}

\newcommand{\rE}{\mathrm{E}}
\newcommand{\rG}{\mathrm{G}}
\newcommand{\halpha}{\hat{\alpha}}
\newcommand{\hbeta}{\hat{\beta}}
\newcommand{\hgamma}{\hat{\gamma}}
\newcommand{\hdelta}{\hat{\delta}}
\newcommand{\hepsilon}{\hat{\epsilon}}
 \newcommand{\myline}{\frac{\ \quad}{\quad}}
  \newcommand{\Z}{\mathbb Z}

\newcommand{\s}{\bm{s}} 
\newcommand{\mm}{\bm{m}}  
\newcommand{\C}{\mathbb{C}}   
\newcommand{\hI}{\hat{I}} 
 
\newcommand{\mat}[4]{ \left( \begin{smallmatrix}   #1 & #2 \\  #3 & #4 \end{smallmatrix}\right) }
\newcommand{\row}[2]{ \left( \begin{smallmatrix}   #1 & #2 \end{smallmatrix}\right) }
\newcommand{\col}[2]{ \left( \begin{smallmatrix}   #1 \\ #2 \end{smallmatrix}\right) }

\begin{document}
\title{Straightened law for quantum isotropic Grassmannian $\OGr^+(5,10)$ } 
\author{M.V. Movshev \\Stony Brook University\\Stony Brook, NY 11794-3651, USA}
\date{\today}
\maketitle
\begin{abstract}
Projective embedding of an  isotropic Grassmannian $\OGr^+(5,10)$  into projective space of spinor representation $S$ can be characterized with a help of $\Gamma$-matrices by equations $\Gamma_{\alpha\beta}^i\lambda^{\alpha}\lambda^{\beta}=0$. A polynomial function of degree $N$ with values in $S$ defines a map to $\OGr^+(5,10)$ if its coefficients satisfy a $2N+1$  quadratic equations. Algebra generated by coefficients of such polynomials is a coordinate ring of the quantum isotropic Grassmannian. We show that this ring  is based on a lattice; its defining  relations satisfy straightened law. This enables us to compute  Poincar\'{e} series of the ring.
\end{abstract}

\tableofcontents

\section{Introduction}\label{S:Intro}
Complex Isotropic Grassmannian $\OGr^+(5,10)$ or a space of pure spinors, as it is known in the physics literature, is a cornerstone of manifestly Poincar\'{e} covariant formulation of string theory \cite{Berkovits}.  A  rigorous construction of the Hilbert space of string theory in the formalism of pure spinors remains a challenging problem (cf. \cite{AABN}). 
Paper \cite{AA}  investigates   a simplified model with a quadric as  the target.  In \cite{Movq} we proved rigorously most of the results of  \cite{AA}. 

The paper \cite{AA} concerns certain algebra-geometric properties of the space of smooth maps $\Map(S^1,Q)\subset \Map(S^1,V^{2n})$ from a circle to a nondegenerate affine quadric $Q$ in $2n$-dimensional linear space $V^{2n}$.  In \cite{Movq} we observed that  analysis of \cite{AA} becomes rigorous if we replace $\Map(S^1,Q)$ by the space of polynomial maps 
 \[\sum_{N\leq k\leq N'}\sum_{\s\in \rG_n}g^k_{\s}v_{\s}z^k\]
 written in some basis $(v_{\s}), \s\in \rG_n=\{1,\dots,n,1^*,\dots,n^*\}$ in $V^{2n}$, and then pass to a  limit $N'\rightarrow \infty, N\rightarrow -\infty$.  In this basis the  $\SO(2n)$-invariant quadratic form $q$ splits: $q=\sum_{i=1}^nx_ix_{i^*}$. An important technical observation, on which hinge all other results in \cite{Movq}, is that the algebra generated by $g^k_{\s}$ is the algebra with straightened law.

 Recall (cf.\cite{EisenbudStr}) that an  algebra $A$ over a ring $R$ with straightened law is based on a directed graph $F$ without oriented cycles. Generators $\tau^{\alpha}$ of the algebra are labelled by vertices of the graph.  The set of vertices is partially ordered: $\alpha\leq \beta$ if there is a directed path  or a  {\it chain} 
\begin{equation}\label{E:chaindef}
\alpha\rightarrow \cdots\rightarrow \beta
\end{equation} in  $F$. 
 $A$ is an  algebra with  straightened law if:
\begin{enumerate}
\item
 {\it Standard monomials} $\tau^{\alpha_1} \cdots \tau^{\alpha_n}$ labelled by $\alpha_1\leq \dots\leq \alpha_n$  form a basis in $A$.
 \item If $\alpha$ and $\beta$ are incomparable and 
 \begin{equation}\label{E:strrelgeneral}
 \tau^{\alpha}\tau^{\beta}=\sum_{i}r_i\tau^{\gamma_{i,1}}\cdots \tau^{\gamma_{i,k_i}}
 \end{equation}
 where $0\neq r_i\in R$, ${\gamma_{i,1}}\leq \cdots \leq {\gamma_{i,k_i}}$ is the unique expression of $\tau^{\alpha}\tau^{\beta}$ in $A$ as a linear combination of standard monomials, then $\gamma_{i,1}\leq \alpha,\beta$.
 \end{enumerate}
 Throughout the paper the ground ring $R$ is the field of complex numbers.
 The goal of the present paper is to establish  straightened law 
 for Quantum Isotropic Grassmannian  in terminology  of \cite{SottileSturmfels} or  the space of Drinfeld's quasimaps of $\P^1$  to $\OGr^+(5,10)$ in terminology of \cite{Braverman}.

 The most straightforward way to define Quantum Isotropic Grassmannian is through $\Gamma$-matrices $\Gamma^{\s}_{\alpha\beta}$.
They are defined  as matrix coefficients    of $\Spin(10)$-intertwiners \begin{equation}\label{E:gammamap}\Sym^2S\rightarrow V\end{equation} 
in a basis 
\begin{equation}\label{E:basis}
\{\theta_{\alpha}|\alpha \in \rE\}
\end{equation} of a spinor representation $S$ and a  basis 
\begin{equation}\label{E:basisvec}
\{v_{\s}| \s\in\rG \}\quad \rG=\rG_5=\{1,\dots,5,1^*,\dots,5^*\}
\end{equation} 
in the defining representation  $V=V^{10}$. Presently we need to know only that cardinality of $\rE$ is equal to sixteen.
By definition 
\begin{equation}\label{E:gammadef}
\sum_{\s\in \rG}\Gamma^{\s}_{\alpha\beta}v_{\s}=\Gamma(\theta_{\alpha},\theta_{\beta}).
\end{equation} 
Generators of algebra  $A_N^{N'}, N\leq N'$ of homogeneous functions on Quantum Isotropic Grassmannian can be arranged into  generating functions 
\begin{equation}\label{E:genuuu}
\lambda^{\alpha}(z)=\sum_{N\leq l\leq N'}\lambda^{\alpha^l}z^l, \quad\alpha^l\overset{\ddef}{=}(\alpha,l)\in \rE\times \mathbb{Z}\overset{\ddef}{=}\hat{\rE}.
\end{equation}  The generating function of relations  $\Gamma^{\s^l}$ is 
\begin{equation}\label{E:relgenfunction}
\sum_{l} \Gamma^{\s^l}z^l=\Gamma^{i}_{\alpha\beta}\lambda^{\alpha}(z)\lambda^{\beta}(z)\quad  \s^l\overset{\ddef}{=} (\s,l)\in  \rG\times \mathbb{Z}\overset{\ddef}{=}\hat \rG
\end{equation}
\begin{definition}\label{D:qgras}
Quantum Isotropic Grassmannian is $Proj(A_{N}^{N'})$.
\end{definition}
 
 The study of straightened law phenomenon is a part of Standard Monomial Theory (cf. \cite{SMT}). It has been  established for algebras of homogeneous functions on Schubert cells  in partial flag spaces of semisimple groups. Straightened law has also been established \cite{LakshmibaiLittelmann} for intersections of Schubert varieties and opposite Schubert varieties (called Richardson varieties). Works \cite{Littelmann1}, \cite{Littelmann2}, and \cite{Littelmann3} generalize results of classical Standard Monomial Theory  to   symmetrizable Kac-Moody algebras. An algebra $C$ of this loosely defined class  enjoys the following properties:
 \begin{enumerate}
 \item\label{a} Straightened law holds.
 \item\label{b} Koszul property holds.
 \item \label{c}$Spec(C)$  is a reduced, irreducible, and normal scheme with   Cohen-Macaulay  singularities.
 \end{enumerate}
  The objects that we are studying in this paper  can be interpreted as  Richardson varieties in a semi-infinite partial flag space of $\hat{\so}_{10}$(cf. \cite{Braverman}),  which are closely related to the spaces   discussed in \cite{FF} and \cite{FFKM}. 
 It seems reasonable to expect that the suitably generalized technique of \cite{LakshmibaiLittelmann} would make it possible  to prove  the listed above package of properties  for spaces of quasimaps to an arbitrary quotient $G/P$ with a semisimple $G$ and a parabolic $P$.

 We however pursue  a more modest goal to establish \ref{a},\ref{b} and some of \ref{c}  for quantum $\OGr^+(5,10)$.  Our analysis relies  on a set of identities between coefficients $\Gamma_{\alpha\beta}^{\s}$ in relations (\ref{E:relgenfunction}). These  Fierz identities are ubiquitous in gauge, gravity, and string theories. 
 
 The diagram $\hat{\rE}$ is fundamental  for   Quantum Isotropic Grassmanians:
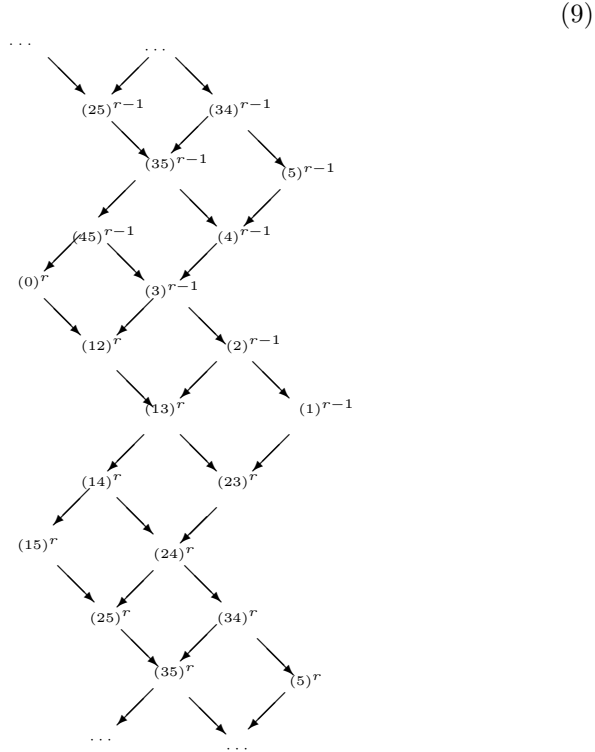
\begin{equation}\label{P:kxccvgw13}
\mbox{
\setlength{\unitlength}{3000sp}
\begin{picture}(5025,-460)(-201,-1600)
\thinlines
\put(2326,-3886){\makebox(0,0)[lb]{\tiny{$(3)^{r-1}$}}}
\put(1537,-1897){\vector( 1,-1){300}}
\put(2063,-2423){\vector( 1,-1){300}}
\put(2851,-2386){\vector(-1,-1){300}}
\put(2363,-1898){\vector(-1,-1){300}}
\put(2588,-1898){\vector( 1,-1){300}}
\put(3188,-2498){\vector( 1,-1){300}}
\put(2626,-2986){\vector( 1,-1){300}}
\put(2700,-3960){\vector( 1,-1){300}}
\put(1501,-3886){\vector( 1,-1){300}}
\put(1876,-5461){\vector(-1,-1){300}}
\put(1612,-6097){\vector( 1,-1){300}}
\put(2138,-6623){\vector( 1,-1){300}}
\put(2926,-6586){\vector(-1,-1){300}}
\put(3488,-7148){\vector(-1,-1){300}}
\put(2663,-6098){\vector( 1,-1){300}}
\put(3263,-6698){\vector( 1,-1){300}}
\put(2701,-7186){\vector( 1,-1){300}}
\put(2401,-6136){\vector(-1,-1){300}}
\put(3226,-4411){\vector( 1,-1){300}}
\put(2926,-5611){\vector(-1,-1){300}}
\put(3526,-5011){\vector(-1,-1){300}}
\put(3451,-2986){\vector(-1,-1){300}}
\put(2401,-7111){\vector(-1,-1){300}}
\put(2926,-4411){\vector(-1,-1){300}}
\put(1801,-3361){\vector(-1,-1){300}}
\put(2926,-3436){\vector(-1,-1){300}}
\put(2251,-2911){\vector(-1,-1){300}}
\put(2026,-3436){\vector( 1,-1){300}}
\put(2401,-3886){\vector(-1,-1){300}}
\put(2626,-5011){\vector( 1,-1){300}}
\put(2101,-5536){\vector( 1,-1){300}}
\put(2326,-5011){\vector(-1,-1){300}}
\put(2101,-4486){\vector( 1,-1){300}}
\put(2326,-1861){\makebox(0,0)[lb]{\tiny{$\cdots$}}}
\put(1801,-2386){\makebox(0,0)[lb]{\tiny{$(25)^{r-1}$}}}
\put(3451,-2911){\makebox(0,0)[lb]{\tiny{$(5)^{r-1}$}}}
\put(2851,-2386){\makebox(0,0)[lb]{\tiny{$(34)^{r-1}$}}}
\put(1201,-1786){\makebox(0,0)[lb]{\tiny{$\dots$}}}
\put(2926,-3436){\makebox(0,0)[lb]{\tiny{$(4)^{r-1}$}}}
\put(2326,-2836){\makebox(0,0)[lb]{\tiny{$(35)^{r-1}$}}}
\put(3001,-4336){\makebox(0,0)[lb]{\tiny{$(2)^{r-1}$}}}
\put(1276,-3811){\makebox(0,0)[lb]{\tiny{$(0)^r$}}}
\put(1801,-4336){\makebox(0,0)[lb]{\tiny{$(12)^r$}}}
\put(2326,-4861){\makebox(0,0)[lb]{\tiny{$(13)^r$}}}
\put(2926,-5461){\makebox(0,0)[lb]{\tiny{$(23)^r$}}}
\put(1801,-5461){\makebox(0,0)[lb]{\tiny{$(14)^r$}}}
\put(2401,-6061){\makebox(0,0)[lb]{\tiny{$(24)^r$}}}
\put(2926,-6586){\makebox(0,0)[lb]{\tiny{$(34)^r$}}}
\put(3526,-7111){\makebox(0,0)[lb]{\tiny{$(5)^r$}}}
\put(2401,-7036){\makebox(0,0)[lb]{\tiny{$(35)^r$}}}
\put(1276,-5986){\makebox(0,0)[lb]{\tiny{$(15)^r$}}}
\put(1876,-6586){\makebox(0,0)[lb]{\tiny{$(25)^r$}}}
\put(3001,-7636){\makebox(0,0)[lb]{\tiny{$\cdots$}}}
\put(3601,-4861){\makebox(0,0)[lb]{\tiny{$(1)^{r-1}$}}}
\put(1876,-7561){\makebox(0,0)[lb]{\tiny{$\cdots$}}}
\put(1726,-3436){\makebox(0,0)[lb]{\tiny{$(45)^{r-1}$}}}
\end{picture}
}
\end{equation}  
\bigskip
\bigskip
\bigskip
\bigskip
\bigskip
\bigskip
\bigskip
\bigskip
\bigskip
\bigskip
\bigskip
\bigskip
\bigskip
\bigskip
\bigskip
\bigskip
\bigskip
\bigskip
\bigskip
\bigskip
\bigskip
\bigskip
\bigskip
\bigskip
\bigskip
\bigskip
\bigskip
\bigskip
\bigskip
\bigskip
\bigskip
\bigskip
\bigskip

To make formulas more readable we denote elements of $\hat \rE$ by Greek letters $\halpha,\dots,\hdelta$ with circumflex accent ,  elements of $\hat \rG$ by similarly accented bold Roman letters $\hat \s, \hat{\bm{t}}$.
Condition (\ref{E:chaindef}) defines a partial order on $\hat{\rE}$.
A segment $[\hdelta,\hdelta']$ is a subset $\{\halpha\in \hat{\rE}|\hdelta\leq \halpha\leq \hdelta'\}$. A set of generators $\lambda^{\halpha}$ of  algebra $A_{N}^{N'}$ is labelled by $\halpha\in [(0)^N, (1)^{N'}], N\leq N'$. It is useful to consider  a more general class of algebras $A_{\hdelta}^{\hdelta'}$ labelled by an interval $[\hdelta,\hdelta']\subset \hat{\rE}$. We shall refer to $Proj(A_{\hdelta}^{\hdelta'})$ as a (semi-infinite) {\it Richardson variety}.

 In this paper we establish that straightened law for $A_{\hdelta}^{\hdelta'}$ (Proposition \ref{P:genstrlaw}) can be derived from the Fierz identity (see Proposition \ref{P:Fierz}). As an immediate corollary  of straightened law  we obtain Cohen-Macaulay property of $A_{\hdelta}^{\hdelta'}$ (Corollary \ref{C:CohenMacaulay}), reducibility of $Spec(A_{\hdelta}^{\hdelta'})$ (Proposition \ref{P:reduce}), and Koszul property (Proposition \ref{P:quadbasiscaf}).

Group of automorphisms of algebras  $A_{\hdelta}^{\hdelta'}$ contains a maximal complex torus $\T^5\subset \Spin(10)$ and $\C^{\times}$(acts as a loop rotation). The characters of  $\T^5\times\C^{\times}$ action in graded components of  $A_{\hdelta}^{\hdelta'}$ can be arranged into  Poincar\'{e} series 
\begin{equation}\label{E:poincare}
A_{\hdelta}^{\hdelta'}(z,q,t), z\in \T^5,q\in \C^{\times}.
\end{equation}
   The formula (\ref{E:genchar}) for  $A_{\hdelta}^{\hdelta'}(z,q,t)$ can be  obtained as a corollary of straightened law. Let $J$ be equal to  
   \begin{equation}\label{E:J}
   \{\dots,\hdelta_0,\hdelta_1,\hdelta_2,\hdelta_3,\hdelta_4,\hdelta_5,\dots,\}=\{\dots,(15),(5),(0)^1,(1),(15)^1,(5)^1,\dots\}
   \end{equation}
   The functions $B_r(t)=A_{(0)}^{\hdelta_r}(1,1,t), r\geq 0$ satisfy recursion (see Section \ref{S:specialcases})
\begin{equation}\label{E:Delan}
\begin{split}
&B_r(t)=\frac{1+t}{(1-t)^2}B_{r-1}(t)+\frac{t}{(1-t)^4}B_{r-2}(t)\\
& B_0(t)=\frac{1}{(1-t)^5},B_1(t)=\frac{1+t}{(1-t)^7}\end{split}
\end{equation}
This gives an effective tool for computation of $B_r$ and is closely related to recursion 
for  Delannoy  numbers (cf.\cite{WDelannoy}). In fact, author came across these formulas during computer experiments with $A_{(0)}^{\hdelta_r}(1,1,t)$ and one of the goals of the present paper was an explanation of these findings. 

Several interesting topics have been left outside of the scope of present paper: firstly, the analysis of limiting characters $A_{\hdelta}^{\hdelta'}(z,q,t), \hdelta\rightarrow \infty$; secondly, the construction of the Hilbert space of $\beta\gamma$-system on pure spinor. These will await a future publication.

The main results of the present paper are Corollaries \ref{P:genstrlaw} \ref{P:quadbasiscaf}, Proposition \ref{P:reduce}, \ref{P:depth}  and formulas (\ref{E:genchar}), (\ref{E:Delan}).  Propositions \ref{P:rela}, \ref{P:obstr}  are technically central in this paper.

The paper is organized as follows. The principal construction, the Hasse diagram based on set of  weights $\rE$ of spinor representation, is introduced in Section \ref{S:finite}. For readers convenience we sketch the proof of straightened law in the case of ordinary pure spinors with an emphasis on Fierz identities. 
 The proof  is done in Section \ref{S:str}.  Combinatorial derivation of the character of $\OGr^+(5,10)$  is given in Section \ref{S:Poincare}.
Algebras of Richardson varieties
are defined in Section  \ref{S:Richardson}.
Section  \ref{S:posetaf} contains some preliminary results on the structure of   a partially ordered set $\hat{\rE}$.
Straightened law for quantum $\OGr^+(5,10)$  and for $Spec(A_{\hdelta}^{\hdelta'})$  are proved in Section \ref{S:alghat}.   Cohen-Macaulay property,  dimension, and depth  of these varieties are also established  Section \ref{S:alghat}. A formula for character $A_{\hdelta}^{\hdelta'}(z,q,t)$ is established in Section \ref{S:genfunaff}. Note  that  Sections \ref{S:genfunaff} and \ref{S:alghat} are logical continuation of Sections \ref{S:str} and \ref{S:Poincare}.  
 Section \ref{S:Delannoyp} contains relevant facts about Delannoy numbers. Appendix contains some technical information. In Appendix \ref{S:Aut} we justify existence of  automorphisms that are needed in the proof of straightened law.
Appendix \ref{A:obstr} contains a list of obstructions that is used in the proof of the main result.
\section{Acknowledgments}
The author would like to thank A.S. Schwarz for useful comments.
 Parts of this work has been written during the authors stay at IHES and Max Plank Institute for Mathematics (Bonn).
The author is grateful to these institutions for excellent working conditions. 
\section{Partial order on the basis of spinor representation $S$}\label{S:finite}
In this section we define  a  partial order in the basis of 16-dimensional spinor representation $S$ of the complex algebraic group $\Spin(10)$, the group of type $D_5$ in Cartan classification. This order is fundamental for the definition of the straightened law.  

 Let $V$ be a fundamental complex 10-dimensional vector representation of complex $\SO(10)$. It carries an $\SO(10)$-invariant inner product $(\cdot ,\cdot)$, the polarization of the quadratic form $q$. We choose a decomposition  $W+W'=<v_1,\dots,v_5>+<v_{1^*},\dots,v_{5^*}>=V$. Subspaces $W$ and $W'$ are isotropic and the inner product satisfies 
 \begin{equation}\label{E:innerproduct}
 (v_i,v_{j^*})=\delta_{ij}.
 \end{equation}
Recall (see \cite{Cartan} for details) that  the spinor representation admits a construction  of a fermionic Fock space. The spinor representation $S$ can be identified with the direct sum 
\begin{equation}\label{E:spinordef}
S=\bigoplus_{i\geq0} \Lambda^{2i} W.
\end{equation}
Tensors 
 $\theta_{(ij)}=v_{i}\wedge v_{j}, \theta_{(k)}=v_{1}\wedge \cdots  \hat{v}_{k} \cdots \wedge v_{5}$ and a constant $1=\theta_{(0)}$ define a bases   in $S$. We set (cf. \ref{E:basis})
\begin{equation}\label{E:zerorelations}
 \rE=\{(0),(ij),(k)|1\leq i<j\leq 5,1\leq k \leq 5\}
\end{equation}
By abuse of notation we shall denote by $\rE$ the following Hasse diagram, whose vertices are labelled by weights of $S$:
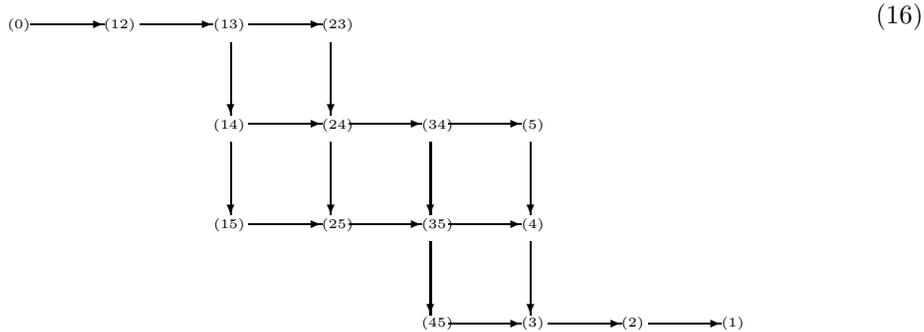
 \begin{equation}\label{P:kxccvgw1}
\mbox{
\setlength{\unitlength}{3000sp}
\begin{picture}(5025,-460)(2101,-1700)
\thinlines
\put(1876,-1711){\vector( 1, 0){600}}
\put(2776,-1711){\vector( 1, 0){600}}
\put(2626,-1861){\vector( 0,-1){600}}
\put(3451,-1861){\vector( 0,-1){600}}
\put(2776,-2536){\vector( 1, 0){600}}
\put(3601,-2536){\vector( 1, 0){600}}
\put(4426,-2536){\vector( 1, 0){600}}
\put(2626,-2686){\vector( 0,-1){600}}
\put(4276,-2686){\vector( 0,-1){600}}
\put(5101,-2686){\vector( 0,-1){600}}
\put(3451,-2686){\vector( 0,-1){600}}
\put(2776,-3361){\vector( 1, 0){600}}
\put(3601,-3361){\vector( 1, 0){600}}
\put(4426,-3361){\vector( 1, 0){600}}
\put(4276,-3511){\vector( 0,-1){600}}
\put(5101,-3511){\vector( 0,-1){600}}
\put(4426,-4186){\vector( 1, 0){600}}
\put(5251,-4186){\vector( 1, 0){600}}
\put(6076,-4186){\vector( 1, 0){600}}
\put(780,-1770){\makebox(0,0)[lb]{\tiny{$(0)$}}}
\put(1576,-1770){\makebox(0,0)[lb]{\tiny{$(12)$}}}
\put(2480,-1770){\makebox(0,0)[lb]{\tiny{$(13)$}}}
\put(3376,-1770){\makebox(0,0)[lb]{\tiny{$(23)$}}}
\put(2480,-2600){\makebox(0,0)[lb]{\tiny{$(14)$}}}
\put(3376,-2600){\makebox(0,0)[lb]{\tiny{$(24)$}}}
\put(4201,-2600){\makebox(0,0)[lb]{\tiny{$(34)$}}}
\put(5026,-2600){\makebox(0,0)[lb]{\tiny{$(5)$}}}
\put(2480,-3420){\makebox(0,0)[lb]{\tiny{$(15)$}}}
\put(3376,-3420){\makebox(0,0)[lb]{\tiny{$(25)$}}}
\put(4201,-3420){\makebox(0,0)[lb]{\tiny{$(35)$}}}
\put(5026,-3420){\makebox(0,0)[lb]{\tiny{$(4)$}}}
\put(4201,-4240){\makebox(0,0)[lb]{\tiny{$(45)$}}}
\put(5026,-4240){\makebox(0,0)[lb]{\tiny{$(3)$}}}
\put(5851,-4240){\makebox(0,0)[lb]{\tiny{$(2)$}}}
\put(6676,-4240){\makebox(0,0)[lb]{\tiny{$(1)$}}}
\put(976,-1711){\vector( 1, 0){600}}
\end{picture}
}
\end{equation}  
\bigskip
\bigskip
\bigskip
\bigskip
\bigskip
\bigskip
\bigskip
\bigskip
\bigskip

  A spinor $\theta\in S$ can be written  as a sum
\[\theta=\lambda\theta_{(0)}+\sum_{i<j} w_{ij}\theta_{(ij)}+\sum_{k=1}^5 p_k \theta_{(k)}=\sum_{\alpha\in \rE}\lambda^{\alpha}\theta_{\alpha}\]

The Lie algebra $\gl_5\subset \so_{10}$ acts tautologically on $W$, via contragradient representation on $W'$ and diagonally on $W+W'$. 
 Vectors $(\theta_{\alpha})$ define  a weight basis for  the algebra of diagonal matrices  
 \begin{equation}\label{E:cartan}
 \h\subset\gl_5\subset \so_{10}
 \end{equation} 
 The action of the  subgroup of diagonal matrices $\widetilde{H}\subset \widetilde{\GL}(5)\subset \Spin(10)$ ( $\widetilde{ }$ stands for two-sheeted cover) on the weight vectors 
 is 
\begin{equation}\label{E:weights}
\begin{split}
&\rho(z)\theta_{(0)} ={\det}^{-\frac{1}{2}}(z)1\overset{\ddef}{=}e_{(0)}(z)1\\
&\rho(z))\theta_{(ij)}={\det}^{-\frac{1}{2}}(z)z_iz_j\theta_{(ij)}\overset{\ddef}{=}e_{(ij)}(z)\theta_{(ij)}\\
&\rho(z))\theta_{(k)}={\det}^{\frac{1}{2}}(z)z^{-1}_k\theta_{(k)}\overset{\ddef}{=}e_{(k)}(z)\theta_{(k)}
\end{split}
\end{equation}
In view of this we can identify $\rE$ with the set of $H$-weights in $S$.

Following Section 7 in \cite{SMT}  we identify  the Weyl group $W(D_5)$  with  a subgroup of a semidirect product $S_5\ltimes (\Z_2)^5$. The symmetric group acts by permutation of entries of the array $(\epsilon_1,\dots,\epsilon_5)\in (\Z_2)^5$. The group $W(D_5)$ is a semidirect product $S_5\ltimes N$, with $N\cong (\Z_2)^4$ being a kernel of the map $(\Z_2)^5\rightarrow \Z_2$ 
\begin{equation}\label{E:split}
(\epsilon_1,\dots,\epsilon_5)\rightarrow \sum_{i=1}^5\epsilon_i.
\end{equation} 
Representation $S$ is {\it minuscule}, i.e. $W(D_5)$ acts transitively on $\rE$ (\cite{SMT}).  The $W(D_5)$ orbit of the highest vector identifies with  the coset $S_5\ltimes N/S_5$. 
The  elements \[(0,0,0,0,0),(0\dots,\overset{i}{1},\dots,0,\dots,\overset{j}{1},\dots,0), (1,\dots,\overset{k}{0},\dots,1) \] in $S_5\ltimes N/S_5\cong (\Z_2)^4\subset (\Z_2)^5$ correspond to elements  $(0),(ij),(k)$ in $\rE$

Let $\sigma_{ij}\in S_5$ be a permutation, $r_{\mm}$ be an operator of multiplication on $\mm\in N$. We shall be using the following set of Coxeter generators $s_1=\sigma_{12}$,$s_2=r_{(1,1,0,0,0)}$, $s_3=\sigma_{23}$, $s_4=\sigma_{34}$, $s_5=\sigma_{45}$, which can be arranged into the Coxeter graph (see \cite{Humphreys} for details):

\[
 \renewcommand{\arraycolsep}{0.1cm}
 \renewcommand{\arraystretch}{1.4}
 \begin{matrix}
 s_{1} &\myline& s_{3} &\myline& s_{4} &\myline& s_{5} \\
&&\big|&&  \\
 && \kern-1cm s_2 \kern-1cm&&
 \end{matrix}
 \] 

The {\it poset} (partially ordered set) $\rE$ with the  order (\ref{E:chaindef}) determined by the diagram (\ref{P:kxccvgw1}) is a lattice in a sense of \cite{Birkhoff} : any two elements $a,b$ have a {\it unique} supremum $a\vee b$ (join) and infimum $a\wedge b$ (meet).
The set of unordered pairs of  non comparable  elements in the poset $\rE$, which determine commutative squares in $\rE$,  is 
\begin{equation}\label{E:Enoncomp}
\begin{split}
&M=\{((14),(23)),((15),(23)),((15),(24)),((15),(34)),((15),(5)),\\
&((25),(34)),((25),(5)),((35),(5)),((45),(5)),((45),(4))\}
\end{split}
\end{equation}
Let $P$ be polynomial algebra $\C[\lambda^{\beta}], \beta\in \rE$. Elements of $M$ as well as monomials $\{\lambda^{\alpha}\lambda^{\alpha'}\in P| (\alpha,\alpha')\in M\}$  shall be called {\it clutters}.
\begin{proposition}\label{E:transitiveM}
Elements of $M$ belong to one $W(D_5)$-orbit in the symmetric square $\Sym^2\rE$.
\end{proposition}
\begin{proof}
We use $W(D_5)$-equivariant identification $\psi:\rE\cong N$. The sum 
\begin{equation}\label{E:mapl}
\begin{split}
&l:\Sym^2\rE\rightarrow N\\
& l(\alpha,\beta)=\psi(\alpha)+\psi(\beta)
\end{split}
\end{equation}
defines $W(D_5)$-map of sets. Fibers of $l$ are $N$-orbits.   The set $l(M)$ is \[\{(0,1,1,1,1),(1,0,1,1,1)(1,1,0,1,1)(1,1,1,0,1)(1,1,1,1,0)\}.\] It is an $S_5$-orbit.
\end{proof}

The graph $\rE$ has group theoretic and Lie-algebraic interpretations. 
The group theoretic approach gives a non-directed graph. Its based on the following    general construction.
\begin{definition}\label{D:graph}
Let $X$ be a set equipped with a  group action. Suppose the group $G$ is generated by $s_1,\dots,s_k$.
The set of vertices of  graph $Q(X,s_1,\dots,s_k)$ is $X$. Two vertices $a\neq b$ in $X$ are connected by an edge in $Q(X,s_1,\dots,s_k)$ if $b=s_ia$ for some $s_i$. In case $b=s_ia$ for two values of $i$  we still have one connecting edge.
\end{definition}
\begin{proposition}
The graph $Q(S_5\ltimes N/S_5,s_1,\dots,s_5)$ is isomorphic to a non-directed version of (\ref{P:kxccvgw1}).
\end{proposition}
\begin{proof}
Straightforward exercise.
\end{proof}

Left multiplication on $v_i$ and contraction with $v^*_j$ define operators in $\Lambda W$. By abuse of notations we denote these operators by $v_i,v^*_j$. Operators $v_iv_j,v_iv^*_j, v^*_iv^*_j, 1$ satisfy commutation relation of $\so_{10}\times \C$ (cf. \cite{ChevalleySpinors}). 
Operators $v_iv^*_i,1, i=1,\dots,5$ define the action of  Cartan subalgebra $\h\subset\gl_5\subset \so_{10}\times \mathbb{C}$ in $S$. 
 Operators in $S$ corresponding to simple positive root vectors $e_i,i=1,\dots,5$ in $\so_{10}$ are 
 \begin{equation}\label{E:R1}
 R^+_2x= v_1v_2x
 \end{equation}
  and 
  \begin{equation}\label{E:R25}
  R^+_ix= v_iv^*_{i-1}x,\quad  i=1,3,4,5.
  \end{equation}

We connect  $\alpha,\alpha'\in \rE$  by a directed edge if    $R^{+}_i\theta_{\alpha}$ is proportional  to $\theta_{\alpha'}$ for some $i$ . 
The choice of $\h$ and $\{e_i\}$ defines a triangular decomposition
\begin{equation}\label{E:so10decomp}
\n_-+\h+\n_+=\so_{10}
\end{equation}
\begin{proposition}
Directed graph obtained this way is isomorphic to $\rE$.
\end{proposition}
\begin{proof}
Direct inspection.
\end{proof}

In the following we shall not make a distinction between directed graphs with without directed cycles and partially ordered sets they define.

  The points of the affine cone defined by equation 
\begin{equation}\label{E:pure}
\Gamma^{\s}\overset{\ddef}{=}\sum_{\alpha\beta\in \rE}\Gamma^{\s}_{\alpha\beta}\lambda^{\alpha}\lambda^{\beta}=0 \quad \s\in \rG
\end{equation} are called pure spinors. It is know \cite{Cartan} that $Proj(A)=\OGr^+(5,10)$ where 
\[A=P/(\Gamma^{\s}) \quad \s\in \rG\]

According to \cite{Berkovits} and  \cite{CortiReid} $\Gamma^{i}$ and  $\Gamma^{i^*}$ are respectively 
\begin{equation}\label{E:relexpl}
\begin{split}
&\lambda p_{i}-\Pf_{i}(w)=0\quad {i}=1,\dots,5\\
&wp=0.
\end{split}
\end{equation}
Last line contains five equations, written in a matrix form.
The function  $\Pf_i(w)$ is the pfaffian of  the complement of $i$-th row and column in $w$.

\section{Proof of the straightened law for algebra $A$}\label{S:str}
Formally this section contains no new results. The goal we pursue here is to familiarize the reader with simplified version of the proof of the main statements in a well studied setting with an emphasis on specifics of a spinor representation.

\begin{proposition}\label{E:degreenull}\cite{SMT} 
\renewcommand{\theenumi}{\roman{enumi}}
\begin{enumerate}
\item Each of the equations (\ref{E:relexpl}) and their expanded version  (\ref{E:equations}) contains a unique clutter.
\item More precisely  equations have  the form 
\begin{equation}\label{E:reluniv}
\lambda^{\alpha}\lambda^{\alpha'}\pm\lambda^{\alpha\vee\alpha'}\lambda^{\alpha\wedge\alpha'}=\sum_{\gamma< \alpha\wedge\alpha',  \gamma'>  \alpha\vee\alpha'} \pm\lambda^{\gamma}\lambda^{\gamma'}
\end{equation}
\end{enumerate}
\end{proposition}

\begin{proof}
We have already mentioned  that several  proofs that work in much more general context are known (see Introduction), we give one proof here with a purpose of generalization in Proposition \ref{P:rela}.

Quadrics (\ref{E:relexpl}) can be written in the expanded form:
\begin{equation}\label{E:equations}
\begin{split}
&\Gamma^{1}=\lambda p_{1} + \underbracket{w_{2, 5} w_{3, 4}} - \overbracket{w_{2, 4} w_{3, 5}} + w_{2, 3} w_{4, 5}\sim 0\\
&\Gamma^{2}=-\lambda  p_{2}-\underbracket{w_{1,5} w_{3,4}}+\overbracket{w_{1,4} w_{3,5}}-w_{1,3} w_{4,5}\sim 0\\
&\Gamma^{3}=\lambda  p_{3}+\underbracket{w_{1,5} w_{2,4}}-\overbracket{w_{1,4} w_{2,5}}+w_{1,2} w_{4,5}\sim 0\\
&\Gamma^{4}=-\lambda  p_{4}-\underbracket{w_{1,5} w_{2,3}}+\overbracket{w_{1,3} w_{2,5}}-w_{1,2} w_{3,5}\sim 0\\
&\Gamma^{5}=\lambda  p_{5}+\underbracket{w_{1,4} w_{2,3}}-\overbracket{w_{1,3} w_{2,4}}+w_{1,2} w_{3,4}\sim 0\\
&\Gamma^{1^*}=-p_{2} w_{1, 2} + p_{3} w_{1, 3} - \overbracket{p_{4} w_{1, 4}} + \underbracket{p_{5} w_{1, 5}}\sim 0\\
&\Gamma^{2^*}=-p_{1} w_{1, 2} + p_{3} w_{2, 3} - \overbracket{p_{4} w_{2, 4}} + \underbracket{p_{5} w_{2, 5}}\sim 0\\
&\Gamma^{3^*}=-p_{1} w_{1, 3} + p_{2} w_{2, 3} - \overbracket{p_{4} w_{3, 4}} + \underbracket{p_{5} w_{3, 5}}\sim 0\\
&\Gamma^{4^*}=-p_{1} w_{1, 4} + p_{2} w_{2, 4} - \overbracket{p_{3} w_{3, 4}} +\underbracket{p_{5} w_{4, 5}}\sim 0\\
&\Gamma^{5^*}=-p_{1} w_{1, 5} + p_{2} w_{2, 5} - \overbracket{p_{3} w_{3, 5}} + \underbracket{p_{4} w_{4, 5}}\sim 0
\end{split}
\end{equation}
A brute force approach to the proof requires verification conditions of the proposition for all equations (\ref{E:equations}).   Monomials marked by underbracket $\underbracket{\ \ }$ are clutters. Monomials  $\lambda^{\alpha\wedge\beta}\lambda^{\alpha\vee\beta}$ are marked by overbracket $\overbracket{\ \ }$.  We leave verifications to the reader.

There is a less computational but lengthier  proof, based on  representation theory. We include it here as a prototype of the proof in the affine case.  

Spinors  decompose under $\so_8\subset \so_{10}$ into a direct sum 
\begin{equation}\label{E:spinorsplit}
S=S^++S^-
\end{equation} of eight-dimensional spinor representations of opposite chiralities (see e.g. \cite{BVinbergALOnishchik}). We choose  Cartan subalgebra   $\h\subset \so_{8}\times \so_{2}\subset \so_{10}$.  $\Gamma$ maps weight vectors $\theta^{\alpha}\theta^{\alpha'}$ in the space of symmetric tensors $ \Sym^2S$ to weight vectors $\Gamma(\theta^{\alpha},\theta^{\alpha'}) \in V$. 
A pair of vectors  $v_+=v_i,v_-=v_{i^*}\in V$  (\ref{E:basisvec}) is  invariant with respect to some  $\so_{8}\subset \so_{10}$ action and get  scaled under action of $\so_{2}\subset \so_{10}$. The remaining  $\{v_j,v_{j^*}|j\neq i\}$ span eight-dimensional vector representation of $\so_{8}$.
 Quadratic functions $\Gamma^{+}=\Gamma^{i},\Gamma^{-}=\Gamma^{i^*}$ define nonzero $\so_{8}$-invariant inner products  $(\cdot,\cdot)_{\pm}$ on  $S^{\pm}$. Note that $\so_2$-scaling of $v_{\pm}$ insures (cf. \ref{E:gammadef}) that  $\Gamma^{+}|_{S^{-}}=\Gamma^{-}|_{S^{+}}=0$. $\Gamma^{\pm}$ are  unique up to a factor (see \cite{BVinbergALOnishchik}). 

Let 
\begin{equation}\label{E:trso8}
\n'_-+\h+\n'_+=(\n_+\cap \so_{8})+(\h\cap \so_{8})+(\n_-\cap \so_{8})
\end{equation} 
the  decomposition of $\so_{8}$ compatible with (\ref{E:so10decomp}).
 Let $\gamma_{\alpha,\alpha'}$ be a path in $\rE$ that connects $\alpha$ and $\alpha'$. It is easy to see that such path exists if and only if  there is an element $u$ in universal enveloping algebra $U(\n_+)$ such that $u\theta^{\alpha}$ is proportional to $\theta^{\alpha'}$.  
 
 We decompose $\rE$ into a union $\rE^{+}\cup \rE^{-}$, with  $\rE^{\pm}=\{\alpha\in \rE|\theta^{\alpha}\in S^{\pm}\}$.  A poset structure on $\rE^{\pm}$ is defined as follows: a directed path $\gamma_{\alpha,\alpha'}$ connects $\alpha, \alpha' \in \rE^{\pm}$ if there $u'\in U(\n_{+}')$ with $u\theta^{\alpha}\sim\theta^{\alpha'}$. In this case we define $\alpha$ to be less or equal to $\alpha'$.  
 
  Simple positive root vectors in $\so_{8}$ with respect to decomposition (\ref{E:trso8})  are $v_2v^{*}_1,v_3v^{*}_2,v_4v^{*}_3, v_1v_2 $ if $v_+=v_5,v_-=v_{5^*}$. Their action on weight vectors in $S^{\pm}\subset S$, which are also $\so_{10}$ weight vectors,  produce the following subdiagram in $\rE$ 
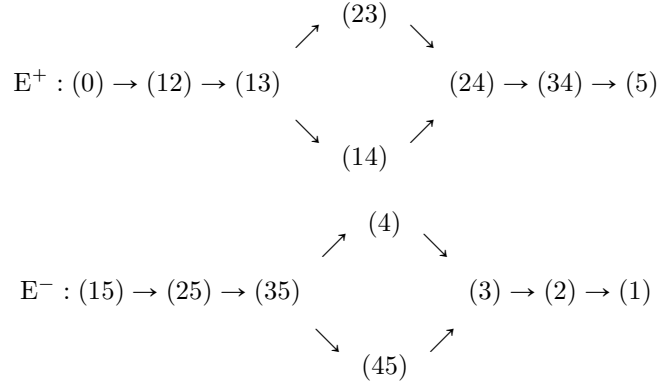
   \[\rE^+: (0) \rightarrow (12)  \rightarrow (13) \begin{array}{c}\nearrow  \begin{array}{c} (23) \\ \text{ } \end{array} \searrow \\ \searrow   \begin{array}{c} \text{ }\\(14)  \end{array} \nearrow \end{array}(24)  \rightarrow (34) \rightarrow (5)  \]
 \[\rE^-: (15) \rightarrow (25)  \rightarrow (35) \begin{array}{c}\nearrow  \begin{array}{c} (4) \\ \text{ } \end{array} \searrow \\ \searrow   \begin{array}{c} \text{ }\\(45)  \end{array} \nearrow \end{array}(3)  \rightarrow (2) \rightarrow (1)  \]
  
  One of the triality automorphisms of  $\so_{8}$ that conjugates $S^{\pm}$ with the defining representation $V^8$ preserves decomposition (\ref{E:trso8}), shuffles simple roots, and conjugate weight vectors and quadratic forms. This is why Hasse diagrams of $S^{\pm}$ and $V^8$ are the same. Hasse diagrams of defining representations of $\so_{2n}$ has been studied in \cite{Movq}. When $n=4$ such diagram is 
  \[ (4) \rightarrow(3)  \rightarrow ({2}) \begin{array}{c}\nearrow  \begin{array}{c} (1) \\ \text{ } \end{array} \searrow \\ \searrow   \begin{array}{c} \text{ }\\(1^*)  \end{array} \nearrow \end{array}(2^*)  \rightarrow (3^*) \rightarrow (4^*)  \]
  
Quadratic form    on $V^8$
  $\sum_{i=1}^4x_ix^{i^*}$ in the basis $(v_{\s})$ contains a unique clutter $x_1x^{1^*}$. Therefore the conjugated $\Gamma^+$ contains a unique  clutter $\lambda^{14}\lambda^{23}$ and   the conjugated $\Gamma^-$ contains a unique clutter $\lambda^{4}\lambda^{45}$.

  The group $W(D_5)$ acts transitively on weights of $V$. By Proposition (\ref{E:transitiveM})  all clutters belong to one $W(D_5)$-orbit. The Weyl group shuffles the quadrics $\Gamma^{\bm{s}}$. The $w$-transformed quadric $\Gamma^{\bm{s}}=(\Gamma^{+})^{ w}, w\in W(D_5)$  contains any given clutter. Ten - the number of quadrics coincides with the number of clutters. A quadric can contain at least one clutter, because an $\rE$-clutter  is an $\rE^{\pm}$-clutter, which is unique as we saw above. Thus, each $\Gamma^{\bm{s}}$ contains a unique $\rE$-clutter.
  
\end{proof}

\begin{corollary}
There is a one-to-one correspondence between a set of clutters $(\alpha,\alpha')\in M\subset  \rE\times \rE$ and quadrics  $\Gamma^{\bm{s}},\bm{s}\in \rG$. 
\end{corollary}
Poincar\'{e} series $A(t)=\sum_{n\geq 0}\dim A_kt^k$ has been determined in \cite{BerNek} using fixed point technique and in  \cite{CortiReid} using method of resolutions.  Our method relies on Gr\"{o}bner basis technique. It admits a straightforward generalization to the algebras $A_{\hdelta}^{\hdelta'}$ (see Section \ref{S:alghat}).

 Relations (\ref{E:reluniv}) provide us with a combinatorial method to compute $A(t)$.
 The method amounts to constructing  a basis  in $A_k\subset A$. Consequently $\dim A_k$ is exactly the number of elements in the basis.  Monomials $\lambda^{\bm{n}}=\prod_{\alpha \in \rE}(\lambda^{\alpha})^{n_{\alpha}}, \deg \lambda^{\bm{n}}=\overline{\bm{n}}=k$ define a basis in $P$, but because of relations in $A$ some inevitably become  redundant in $A_{k}$. Our goal is to construct a subset $B_k\subset \{\lambda^{\bm{n}}|\overline{\bm{n}}=k\}$, which becomes a basis in $A_{k}$. Any monomial $\lambda^{\bm{n}}$, because of the relations (\ref{E:reluniv}), is a linear combination of standard monomials (see Section \ref{S:Intro}).
  
We denote by $X_k$ the  set of standard monomials of degree $k$. The process of elimination of clutters using relations is called reduction. We associate a reduction with  each  $\Gamma^{\s}$ in (\ref{E:equations}).

The set $X_k$ might not be a basis.  Some monomials like $p_5w_{2,5}w_{3,4}$, which are called {\it obstructions},  admit several reductions: we can apply reduction $\Gamma^{2^*}$ to $p_5w_{2,5}$ or $\Gamma^1$ to $w_{2,5}w_{3,4}$.   

We shall use term obstruction also for a set of unordered triples $(\alpha,\beta,\gamma)\in \Sym^3 \rE$ such that $(\alpha,\beta), (\alpha,\gamma)\in M$. They encode obstructing monomials $\lambda^{\alpha}\lambda^{\beta}\lambda^{\gamma}$.
For technical purposes we need to modify slightly this definition.
\begin{definition}\label{D:obstruction}

An   (unordered) list  $(\gamma, (\alpha ,\beta))$ of elements in a poset $\rE$ is called a noncommutative obstruction if $(\alpha,\beta)$ is a clutter and if $(\alpha,\gamma)$ or $(\beta,\gamma)$ is a clutter. Thus noncommutative obstructions come in pairs:
\[\{(\gamma, (\alpha ,\beta)), (\alpha,(\beta,\gamma))\} \mbox{ or }  
\{(\gamma, (\alpha ,\beta)), (\beta,(\alpha,\gamma))\}
\]
By abuse of notations we shall also call a monomial $\lambda^{\gamma} \otimes  \lambda^{\alpha} \lambda^{\beta}$ a noncommutative  obstruction.
\end{definition}
If $(\gamma, (\alpha ,\beta))$ is a noncommutative obstruction, then $(\gamma, \alpha ,\beta)$ is an ordinary obstruction. A list (\ref{E:obstructionsord}) of pairs noncommutative obstructions in $\rE$ is given in Appendix \ref{A:obstr}. 

It is conceivable that different reductions of an obstruction  result in  distinct  non reducible expressions.
It would mean that we need to further reduce the set $X_k$.
Luckily, this is does not happen in our case.
There is a technology of Gr\"{o}bner bases (see \cite{Mora} for introduction and references), which has been specifically designed to deal with such issues. 
In order to apply it we  refine our partial order on $\rE$ to a strict total order:

\begin{equation}\label{E:order}
\begin{split}&(0) <(12)<(13)<(14)<(23)<(15)<(24)<(25)<\\
&<(34)<(35)<(5)<(45)<(4)<(3)<(2)<(1)
\end{split}
\end{equation}

We extend this order  to degree lexicographic order in the monomial basis in  $P$. 
 The tip $\rT(f)=lc(f)\lambda^{\bm{n}}$  is  the maximal monomial of $f\in P$, $lc(f)$ is the leading coefficient.
We write $f>g$ for polynomials $f,g$ with $\rT(f)>\rT(g)$.
The order is chosen so that reductions decrease it:
\[\rT(\Gamma^{\bm{s}})>\rT(\Gamma^{\bm{s}})-\Gamma^{\bm{s}}\]

The "Diamond Lemma" \cite{Bergman} guarantees that if results of any two reduction   of all (cubic) obstructions agree,  then  $X_k$ is a basis in $A_k$.
 To be precise we define $T(f,g)$ to be the least common multiple ($lcm$) of $\rT(f)$ and $\rT(g)$. 
\begin{definition}
The polynomial \[S(f,g)=lc(f)\frac{\rT(f,g)}{\rT(g)}g-lc(g)\frac{\rT(f,g)}{\rT(f)}f\]   is called $S$-polynomial of $f,g$.
\end{definition}
Let  $\lambda^{\bm{n}}$ be some cubic obstruction, $\Gamma^{\bm{s}},\Gamma^{\bm{s}'}$ be a pair of relations in (\ref{E:equations}) such that $lcm(\rT(\Gamma^{\bm{s}}),\rT(\Gamma^{\bm{s}'}))=\lambda^{\bm{n}}$, then $S(\Gamma^{\bm{s}},\Gamma^{\bm{s}'})$ contains no multiple of $\lambda^{\bm{n}}$: it  has been canceled in the subtraction. 

\begin{proposition}\label{P:Buchbergersimp}
If for any pair of generating relations  $\Gamma^{\bm{s}},\Gamma^{\bm{s}'}\in I\subset P$ (${\bm{s}},{\bm{s}'}\in \rG$) with $\rT(\Gamma^{\bm{s}},\Gamma^{\bm{s}'})\neq \rT(\Gamma^{\bm{s}})\rT(\Gamma^{\bm{s}'})$ the $S$-polynomials are  
\begin{equation}\label{E:bush}
S(\Gamma^{\bm{s}},\Gamma^{\bm{s}'})=\sum_{\alpha\in \rE}\sum_{{\bm{r}}\in \rG} c_{\alpha,{\bm{r}}}\lambda^{\alpha}\Gamma^{\bm{r}}
\end{equation} with no obstructions among monomials in $\lambda^{\alpha}\Gamma^{\bm{r}}$ , then the set $X_k$ is a basis in $A_k$.
\end{proposition}
\begin{proof}
This is a trivial restatement of Buchberger algorithm for  a commutative Gr\"{o}bner basis (\cite{Mora}) adapted to our needs. More precisely this is a statement that the algorithm terminates at the first iteration. 
\end{proof}

The inner product (\ref{E:innerproduct}) can be used to raise and to lower the $\s$-index in the tensor $\Gamma^{\s}_{\alpha\beta}$. 

\begin{proposition}\label{P:Fierz}
Quadratic functions  $\Gamma^{\s}$ satisfy
\begin{equation}\label{E:Fierz}
I_{\alpha}=\sum_{\s\in\rG}\Gamma^{\s}_{\alpha\beta}\lambda^{\beta}\Gamma_{\s}=0
\end{equation}
\end{proposition}
\begin{proof}
This is equivalent to Fierz identity
\[\Gamma^{\s}_{\alpha(\beta}\Gamma_{\gamma\delta)\s}=0.\]
 $()$ stands for symmetrization.

  Theory of  invariants gives a different outlook on this identity. It is not hard to see that $I_{\alpha}d\lambda^{\alpha}$ is a differential of  $I(\lambda)=\Gamma^{\s}\Gamma_{\s}  =\sum_{i=1}^5\Gamma^{i}\Gamma^{i^*}$.  The function $I(\lambda)$ is by construction a $\Spin(10)$-invariant. By Igusa's classification of spinors in ten dimensions \cite{Igusa}  $S$ contains a dense orbit.  
   This implies that $I(\lambda) \equiv I(0) = 0$, hence $I_{\alpha}=0$.
\end{proof}

\begin{proposition}\label{P:quadbasis}Fix  the order (\ref{E:order}) on generators of $P$.
Then assumptions of Proposition \ref{P:Buchbergersimp} are satisfied for relations (\ref{E:equations}).
\end{proposition}
\begin{proof}
The algebra $P$ coincides with symmetric algebra $\Sym\ S^*$. Collection of  Fierz identities describes a basis in the  kernel of the multiplication map $S^*\otimes V^*\rightarrow  \Sym^3 S^* $: 
\begin{equation}\label{E:subst}
 \lambda^{\alpha}x_{\s}\rightarrow \lambda^{\alpha}\Gamma^{\s}
 \end{equation} 
 The kernel as a representation of $\so_{10}$  is isomorphic to $S$ (e.g. \cite{BVinbergALOnishchik}). Its weights are opposite to weights of $S^*$. We use these weights to label individual Fierz identities:
 
{\tiny
\begin{equation}\label{E:fierzksimple}
\begin{split}
&h_{(0)}= p_{1} x_{1^*}-p_{2} x_{2^*}+p_{3} x_{3^*}-\underbracket{p_{4} x_{4^*}}+\underbracket{p_{5} x_{5^*}}\\
&h_{(12)}= \underbracket{w_{4,5} x_{3^*}}-\underbracket{w_{3,5} x_{4^*}}+w_{3,4} x_{5^*}-p_{2} x_{1}-p_{1} x_{2}\\
&h_{(13)}=  -\underbracket{w_{4, 5} x_{2^*}} + \underbracket{w_{2, 5} x_{4^*}} - w_{2, 4} x_{5^*} + p_{3} x_{1} - p_{1} x_{3}\\
&h_{(23)}= \underbracket{w_{4, 5} x_{1^*}} - \underbracket{w_{1, 5} x_{4^*}} + w_{1, 4} x_{5^*} + p_{3} x_{2} + p_{2} x_{3}\\
&h_{(14)}=\ \underbracket{w_{3, 5} x_{2^*}} - \underbracket{w_{2, 5} x_{3^*}} + w_{2, 3} x_{5^*} - p_{4} x_{1} - p_{1} x_{4}\\
&h_{(24)}=  -\underbracket{w_{3, 5} x_{1^*}} + \underbracket{w_{1, 5} x_{3^*}} - w_{1, 3} x_{5^*} - p_{4} x_{2} + p_{2} x_{4}\\
&h_{(15)}= -\underbracket{w_{3, 4} x_{2^*}} + w_{2, 4} x_{3^*} - w_{2, 3} x_{4^*} + \underbracket{p_{5} x_{1}} - p_{1} x_{5}\\
&h_{(34)}= \underbracket{w_{2, 5} x_{1^*}} - \underbracket{w_{1, 5} x_{2^*}} + w_{1, 2} x_{5^*} - p_{4} x_{3} - p_{3} x_{4}\\
&h_{(25)}= \underbracket{w_{3, 4} x_{1^*}} - w_{1, 4} x_{3^*} + w_{1, 3} x_{4^*} + \underbracket{p_{5} x_{2}} + p_{2} x_{5}\\
&h_{(5)}= \lambda  x_{5^*}+\underbracket{w_{1,5} x_{1}}+\underbracket{w_{2,5} x_{2}}+w_{3,5} x_{3}+w_{4,5} x_{4}\\
&h_{(35)}= -\underbracket{w_{2, 4} x_{1^*}} + w_{1, 4} x_{2^*} - w_{1, 2} x_{4^*} + \underbracket{p_{5} x_{3}} - p_{3} x_{5}\\
&h_{(4)}= -\lambda x_{4^*} - w_{1, 4} x_{1} - \underbracket{w_{2, 4} x_{2}} - \underbracket{w_{3, 4} x_{3}} +  w_{4, 5} x_{5}\\
&h_{(45)}= \underbracket{w_{2, 3} x_{1^*}} - w_{1, 3} x_{2^*} + w_{1, 2} x_{3^*} + \underbracket{p_{5} x_{4}} + p_{4} x_{5}\\
&h_{(3)}= \lambda  x_{3^*}+w_{1,3} x_{1}+\underbracket{w_{2,3} x_{2}}-\underbracket{w_{3,4} x_{4}}-w_{3,5} x_{5}\\
&h_{(2)}= -\lambda x_{2^*} - w_{1, 2} x_{1} + \underbracket{w_{2, 3} x_{3}} + \underbracket{w_{2, 4} x_{4}} +  w_{2, 5} x_{5}\\
&h_{(1)}= \lambda x_{1^*} - w_{1, 2} x_{2} - w_{1, 3} x_{3} -\underbracket{ w_{1, 4} x_{4}} - \underbracket{w_{1, 5} x_{5}}
\end{split}
\end{equation}
}
Monomials $\rT(\Gamma^{\bm{s}},\Gamma^{\bm{s}'})$ such that $\rT(\Gamma^{\bm{s}},\Gamma^{\bm{s}'})\neq \rT(\Gamma^{\bm{s}})\rT(\Gamma^{\bm{s}'})$ are  degree three obstructions.
Each $h_{\alpha}$ upon substitution (\ref{E:subst}) contains precisely a one pair of relations $\Gamma^{\bm{s}},\Gamma^{\bm{s}'}$  (the corresponding terms are underlined) with $\rT(\Gamma^{\bm{s}},\Gamma^{\bm{s}'})\neq \rT(\Gamma^{\bm{s}})\rT(\Gamma^{\bm{s}'})$. 
The reader can check  that underscored terms  contain the same $h_{\alpha}$-dependent obstruction monomial. For example, $p_{4} \Gamma^{4^*}$ and $p_{5} \Gamma^{5^*}$ in the first identity contain $p_4p_5w_{45}$. These monomials have  opposite signs and  cancel each other because of that.  No other term  in $h_{\alpha}$ contains an obstruction. A straightforward analysis of the lattice $\rE$ reveals that there are  sixteen obstructions of degree three (see Table \ref{E:obstructionsord}) and each of them appears as  monomial in one of the underlined terms. 

We use these identities to rewrite $S$-polynomials in the  form (\ref{E:bush})

\end{proof}
\begin{corollary}\label{P:straight}(cf.\cite{SMT})$A$ is an algebra with straightened law. In particular standard monomials form a $\T^5$ weight basis.
\end{corollary}
\begin{corollary}\label{P:quadbasisc}(cf. \cite{Ravi} \cite{Bogvad} \cite{Bezr})Fix the order  (\ref{E:order}) on generators of $P$.
Then relations (\ref{E:equations}) form a quadratic Gr\"{o}bner basis(see \cite{PP} for details) in the ideal $I$. $A$ is a Koszul algebra
\end{corollary}
\begin{proof}
Koszul property is a  corollary (see \cite{PP}) of existence of quadratic Gr\"{o}bner basis . 
\end{proof}

We define algebras $A_{\delta}^{\delta'}$ as follows. An ideal $J_{\delta}^{\delta'}\subset A$ is generated by $\{\lambda^{\alpha}| \alpha\in \rE \backslash [\delta,\delta']\}$, then  \[A_{\delta}^{\delta'}=A/J_{\delta}^{\delta'}.\]
Projective spectrum $Proj(A_{\delta}^{\delta'})$ is called a Richardson variety \cite{LakshmibaiLittelmann}. In particular  $Proj(A_{(0)}^{\delta})$ a Schubert variety; $Proj(A_{\delta}^{(1)})$ is the opposite Schubert variety. 
\begin{remark}Relations in $A_{\delta}^{\delta'}$ are one-to-one correspondence with clutters in $[{\delta},{\delta'}]$.
Statements of Corollaries \ref{P:straight} \ref{P:quadbasisc} are valid for algebras $A_{\delta}^{\delta'}$ (\cite{SMT},\cite{Ravi} \cite{Bogvad} \cite{Bezr}). Standard monomials $\lambda^{\alpha_1}\cdots \lambda^{\alpha_n}$ $\alpha_1 \leq\cdots \leq \alpha_n\in [\delta, \delta']$ form a basis in  $A_{\delta}^{\delta'}$.
\end{remark}
\begin{example}\label{E:proj}Algebra $A_{(0)}^{(15)}$ coincides with $\C[\lambda^{0},\lambda^{12},\lambda^{13},\lambda^{14},\lambda^{15}]$ because $[{(0)},{(15)}]$ contains no clutters (cf. \ref{P:kxccvgw1}). The map of algebras $A\rightarrow A_{(0)}^{(15)}$ encodes an embedding $\P^4\subset \OGr^{+}(5,10)$.
\end{example}
\begin{example}\label{E:quad}Algebra $A_{(0)}^{(5)}$ coincides with \[\C[\lambda^{0},\lambda^{12},\lambda^{13},\lambda^{14},\lambda^{24},\lambda^{34},\lambda^{5}]/(\lambda^{0}  \lambda^{5}+{\lambda^{14} \lambda^{23}}-{\lambda^{13} \lambda^{2,4}}+\lambda^{12} \lambda^{34})\] because $[{(0)},{(5)}]$ contains one clutter (cf. \ref{P:kxccvgw1}). This clutter is associated with relation  $\Gamma^5$ (\ref{E:equations}).  The polynomial algebra is isomorphic to symmetric algebra of eight-dimensional spinors $S^+$. Relation defines quadratic form on  $S^+$, which appears in the proof of Proposition \ref{E:degreenull}.  Geometrically projection $A\rightarrow A_{(0)}^{(5)}$ defines an embedding  of a quadric $Q^6$ into $\OGr^{+}(5,10)$.
\end{example}

\section{The Poincar\'{e} series of algebra $A$}\label{S:Poincare}
Poincar\'{e} series $A(t)$ of algebra $A$ depends only on the combinatorics of the poset $\rE$. We choose to define generating function $C(F)(t)$ of a finite poset $F$, which  in case of $\rE$ coincides with $A(t)$.  In this section we also establish   relations between  $C(F)(t)$ and $C(F')(t)$ for some simple subposets $F'\subset F$. These relations define recursive relations between $A_{\delta}^{\delta'}(t)$ for different $\delta$ and  $\delta'$.

To define a generating function function of a poset $F$ we introduce a  polynomial algebra $\C[e_{\alpha}]$, whose generators are labelled by elements of $F$. 
The generating function $C(F)\in \C[e_{\alpha}][[t]]$ does a weighted counting of chains
\[C_k(F)=\{\alpha_1\leq \cdots \leq \alpha_k|\alpha_i\in F\}\] by the formula

\begin{equation}\label{E:Fgenfun}
C(F)(t)=\sum_{k\geq 0}\sum_{\{\alpha_1\leq \cdots\leq \alpha_k\}\in C_k(F)}e_{\alpha_1}\cdots e_{\alpha_k} t^k
\end{equation}
The chains in $\rE$ are one-to-one correspondence with standard monomials, which by Corollary \ref{P:straight} form a $\T^5$-basis in $A$.
Poincar\'{e} series  $A(z,t)$ (\ref{E:poincare}) is some specialization $C(\rE)(t)$. 
The $\T^5$-action on $\lambda^{\alpha}$ is  multiplication on the weight $e_{\alpha}=e_{\alpha}(z)=e_{\alpha}(z_1,\dots,z_5)$ (\ref{E:weights}). Under this specialization  the coefficient of $t^k$ in (\ref{E:Fgenfun}) is the scaling factor of $\T^5$-action  on $\lambda^{\alpha_1}\cdots\lambda^{\alpha_k}$.

More generally $A_{\delta}^{\delta'}(z,t)$ is equal to $C([\delta, \delta'])(t)$ under this  specialization of $e_{\alpha}$.

A bit of terminology:  $F$ is a {\it convenient} lattice if 
 $[\delta, \delta')$ is a union of at most two intervals 
\begin{equation}\label{E:decomp}
[\delta, \delta')=[\delta, \alpha]\cup[\delta, \alpha']
\end{equation}

If $[\delta, \delta')=[\delta, \alpha]$ for some $\alpha$, then $\delta'$ is called a {\it tail} of  $[\delta, \delta']$.
We say that a convenient lattice $F$ is {\it narrow} if for any $\delta'\in F$, as in (\ref{E:decomp}), such that $\alpha\neq \alpha'$  either $\alpha$ or $\alpha'$ is a  tail. 

Direct inspection  of the diagram (\ref{P:kxccvgw1}) shows that lattice $\rE$ is  convenient and  narrow.
\begin{proposition}\label{P:convnar}
Let $F$ be convenient and  narrow finite lattice. Then for any interval $[\delta,\delta']\subset F$ we have a dichotomy:
\begin{equation}\label{E:tail}
C([\delta,\delta'])(t)=C([\delta,\alpha])(t)\frac{1}{(1-e_{\delta'}t)}
\end{equation}
or
\begin{equation}\label{E:reduction}
C([\delta,\delta'])(t)=\frac{1}{(1-e_{\delta'}t)}\left(e_{\alpha}tC([\delta,\alpha])(t)+C([\delta,\alpha'])(t)\right)
\end{equation}
depending on the number of intervals in decomposition (\ref{E:decomp}). The factor $e_{\alpha}t$ corresponds to an interval $[\delta,\alpha]$ in which $\alpha$ is a tail.

For lower bounds we have
\begin{equation}\label{E:lowbound}
\begin{split}
&C([\delta,\delta'])(t)=\frac{1}{(1-e_{\delta}t)}C([\alpha,\delta])(t)
\quad \mbox{ or }\\
&C([\delta,\delta'])(t)=\frac{1}{(1-e_{\delta}t)}\left(e_{\alpha}tC([\alpha,\delta])(t)+C([\alpha',\delta])(t)\right)
\end{split}
\end{equation}
\end{proposition}

\begin{proof}
For any $\delta \leq \delta'$ in $ F$ we have an identification
\[C_k([\delta, \delta'])=\coprod_{k=l+l'} C_l([\delta, \delta'))\times C_{l'}([\delta', \delta'])\]
 Equalities 
\begin{equation}\label{E:factor}
C([\delta,\delta'])(t)=\frac{1}{(1-e_{\delta'}t)}C([\delta,\delta'))(t)=\frac{1}{(1-e_{\delta'}t)}C([\delta, \alpha]\cup[\delta, \alpha'])(t)
\end{equation}
are valid because $F$ is convenient.

$F$ is a lattice therefore
\begin{equation}\label{E:union}
C_{l}([\delta, \alpha]\cup[\delta, \alpha'])=C_{l}([\delta, \alpha])\cup C_{l}([\delta, \alpha'])
\end{equation} 
and 
\begin{equation}\label{E:intersection}
C_{l}([\delta, \alpha])\cap C_{l}([\delta, \alpha'])=C_{l}([\delta, \alpha\wedge \alpha'])
\end{equation} 

If  $\alpha=\alpha'$ then equation (\ref{E:factor}) implies (\ref{E:tail}). If  $\alpha\neq \alpha'$ in (\ref{E:decomp})  then equation (\ref{E:factor},\ref{E:union},\ref{E:intersection}) implies
\begin{equation}\label{E:reductiongen}
C([\delta,\delta'])(t)=\frac{1}{(1-e_{\delta'}t)}\left(C([\delta,\alpha](t)+C([\delta,\alpha'](t)-C([\delta,\alpha\wedge \alpha'])(t)\right).
\end{equation}

Poset $F$ is narrow. 
Let suppose for certainty $\alpha$ in decomposition (\ref{E:decomp}) is a tail of $[\delta, \alpha]$. Then $[\delta, \alpha)=[\delta, \alpha\wedge \alpha']$ and 
$C([\delta,\alpha])(t)=\frac{1}{(1-e_{\alpha}t)}C([\delta,\alpha\wedge \alpha'])(t)$ together with (\ref{E:reductiongen}) imply (\ref{E:reduction}). 

For lower bounds the proof is similar and is omitted.
\end{proof}
\begin{definition}\label{D:rht}
We define an integer-valued function $\rht$ on $\rE$ by the rule: if there is an arrow ${\alpha}\rightarrow {\beta}$, then 
\begin{equation}\label{E:ht}
\rht({\alpha})+1=\rht({\beta}).
\end{equation}
 The function is normalized by condition $\rht((0))=0$. 
 \end{definition}

Note that equations (\ref{E:tail},\ref{E:reduction}) and an obvious normalization
\begin{equation}\label{E:norm}
C([\delta,\delta])(t)=\frac{1}{1-e_{\delta}t}
\end{equation}
can be used for inductive computations of $A_{\delta}^{\delta'}(t)$ because  $\rht(\alpha),\rht(\alpha')= \rht(\delta')-1$.
For example $A_{(0)}^{(1)}(t)$ can be computed as follows:

\[\begin{split}
&A_{(0)}^{(1)}(t)=A_{(0)}^{(2)}(t)\frac{1}{1-e_{(1)}t}=A_{(0)}^{(3)}(t)\frac{1}{1-e_{(2)}t}\frac{1}{1-e_{(1)}t}=\\
&=\row{A_{(0)}^{(45)}(t)}{A_{(0)}^{(4)}(t)}\col{te_{(45)}}{1}\frac{1}{1-e_{(3)}t} \frac{1}{1-e_{(2)}t}\frac{1}{1-e_{(1)}t}=\\
&=\row{A_{(0)}^{(35)}(t)}{A_{(0)}^{(5)}(t)}
\mat{\frac{1}{1-e_{(45)}t}}{\frac{1}{1-e_{(4)}t}}{0}{\frac{te_{(5)}}{1-e_{(4)}t}}
\col{te_{(45)}}{1}\frac{1}{1-e_{(3)}t} \frac{1}{1-e_{(2)}t}\frac{1}{1-e_{(1)}t}=\\
&=\row{A_{(0)}^{(25)}(t)}{A_{(0)}^{(34)}(t)}
\mat{\frac{1}{1-e_{(35)}t}}{0}{\frac{te_{(34)}}{1-e_{(35)}t}}{\frac{1}{1-e_{(5)}t}}
\mat{\frac{1}{1-e_{(45)}t}}{\frac{1}{1-e_{(4)}t}}{0}{\frac{te_{(5)}}{1-e_{(4)}t}}
\col{te_{(45)}}{1}\frac{1}{1-e_{(3)}t} \frac{1}{1-e_{(2)}t}\frac{1}{1-e_{(1)}t}\\
&\cdots\\
&=
\frac{1}{1-e_{(0)}t} \frac{1}{1-e_{(12)}t}\frac{1}{1-e_{(13)}t}
\row{\frac{1}{1-e_{(14)}t}}{\frac{1}{1-e_{(25)}t}}\times\\
&\times\mat{\frac{1}{1-e_{(15)}t}}{\frac{1}{1-e_{(24)}t}}{0}{\frac{te_{(23)}}{1-e_{(24)}t}}
\mat{\frac{te_{(15)}}{1-e_{(25)}t}}{0}{\frac{1}{1-e_{(25)}t}}{\frac{1}{1-e_{(34)}t}}\times\\
&\times\mat{\frac{1}{1-e_{(35)}t}}{0}{\frac{te_{(34)}}{1-e_{(35)}t}}{\frac{1}{1-e_{(5)}t}}
\mat{\frac{1}{1-e_{(45)}t}}{\frac{1}{1-e_{(4)}t}}{0}{\frac{te_{(5)}}{1-e_{(4)}t}}\times\\
&\times\col{te_{(45)}}{1}\frac{1}{1-e_{(3)}t} \frac{1}{1-e_{(2)}t}\frac{1}{1-e_{(1)}t}\\
\end{split}\]

The product upon substitution $e_{\alpha}=1$ becomes $\frac{1+5t+5t^2+t^3}{(1-t)^{11}}$ (cf. \cite{CortiReid},\cite{BerNek}).

\section{Algebras of Richardson varieties  }\label{S:Richardson}
We start this section with an accurate description of algebra $A_{\hdelta}^{\hdelta'}$, whose projective spectrum by definition is a semi-infinite Richardson variety. In the end  we define the inverse limit of  $A_{\hdelta}^{\hdelta'}$.
\begin{definition}\label{D:algAdelta}
The algebra $A_{\hdelta}^{\hdelta'}$ is  a quotient of $\mathbb{C}[\lambda^{\hat{\alpha}}], \hat{\alpha}\in [\hat{\delta},\hat{\delta}']\subset \hat{\rE}$. 
 The ideal of relations is generated by the coefficients $\Gamma^{\hat{\s}}$ of the generating function (\ref{E:relgenfunction})
where 
\begin{equation} \label{E:genfunl}
\lambda^{\alpha}(z)=\sum_{\alpha^l\in [\hat{\delta},\hat{\delta}'] }\lambda^{\alpha^l}z^l 
\end{equation}
\end{definition}

The  set of equations (\ref{E:genuuu},\ref{E:relgenfunction}) define an affine variety that parametrizes  the curves of degree $\leq N'-N$ is the affine cone (\ref{E:pure}). We shall refer to equations  (\ref{E:relgenfunction}) as the {\it affinization} of equations (\ref{E:pure}).

 There is a  surjective homomorphisms $A_{\hdelta}^{\hdelta'}\rightarrow A_{\hgamma}^{\hgamma'}$ $ \hdelta\leq \hgamma \leq \hgamma' \leq \hdelta'$ ,  whose kernel is  generated by $\lambda_{\hat \alpha}\in A_{\hdelta}^{\hdelta'}, \hat \alpha\in [\hdelta,\hdelta']\backslash [\hgamma,\hgamma']$. The homomorphism acts as an identity on the remaining  generators . This   makes collection $ A_{\hdelta}^{\hdelta'}$ an inverse system.

Let $\hat{P}'$ be a polynomial algebra $\mathbb{C}[\lambda_{\hat\alpha}]$, $\hat\alpha\in \hat \rE$. We define a topology  on  $\hat \rE$ by declaring a basis of open sets to be $\{T^l[(0),\infty)|l\in  \mathbb{Z}\}$ and  $\{T^l(-\infty,(0)]|l\in  \mathbb{Z}\}$. 
We denote completion of $\hat{P}'$ with respect to this topology by $\hat P$. It contains a linear space of homogenous elements $\hat{P}_k$ of degree $k$. An ideal $I$ in $\hat P$ topologically generated by expression (\ref{E:relgenfunction}) with no  restrictions on $l,l'$ and $k$. We denote by $\hat A$ the quotient $\hat P/I$. There is a homomorphism $\hat A\rightarrow  A_{\hdelta}^{\hdelta'}$, that acts as an identity on $\lambda_{\hat \alpha}, \hat \alpha\in [\hdelta,\hdelta']\subset \hat \rE$  and by zero on $\hat \rE\backslash [\hdelta,\hdelta']$
\begin{proposition}\label{P:invlim}
There is an isomorphism $e:\hat A\rightarrow \underset{\longleftarrow}\lim A_{\hdelta}^{\hdelta'}$
\end{proposition}
\begin{proof}
We define  $P_{\hdelta}^{\hdelta'}$ to be a polynomial algebra $\mathbb{C}[\lambda_{\hat\alpha}]$, $\hat\alpha\in [\hdelta,\hdelta']$. Then $P_{\hdelta}^{\hdelta'}=\bigoplus_{k\geq 0} P_{\hdelta\ k}^{\hdelta'}$ is a decomposition into graded components. Let $I_{\hdelta}^{\hdelta'}$ be the kernel of the homomorphism $P_{\hdelta}^{\hdelta'}\rightarrow A_{\hdelta}^{\hdelta'}$. The map $r_{\hdelta,\hdelta',\hgamma,\hgamma'}:P_{\hdelta}^{\hdelta'}\rightarrow P_{\hgamma}^{\hgamma'}$ is defined on the generators by the same formulas as the map  $A_{\hdelta}^{\hdelta'}\rightarrow A_{\hgamma}^{\hgamma'}$. The maps $r_{\hdelta,\hdelta',\hgamma,\hgamma'}$ are surjective. Moreover the maps $r_{\hdelta,\hdelta',\hgamma,\hgamma'}:I_{\hdelta}^{\hdelta'}\rightarrow I_{\hgamma}^{\hgamma'}$ are surjective (this can be checked on generators of degree two). The vanishing of $\underset{\longleftarrow}\lim^1 I_{\hdelta}^{\hdelta'}$  follows (see \cite{McCleary} Section 3.2) from this condition.   We have a short exact sequence of limits
\[0\rightarrow \underset{\longleftarrow}\lim I_{\hdelta}^{\hdelta'}\rightarrow \underset{\longleftarrow}\lim P_{\hdelta}^{\hdelta'}\rightarrow \underset{\longleftarrow}\lim A_{\hdelta}^{\hdelta'}\rightarrow 0\]
For this reason  $\hat P= \underset{\longleftarrow}\lim P_{\hdelta}^{\hdelta'}$, $I= \underset{\longleftarrow}\lim I_{\hdelta}^{\hdelta'}$ (see \cite{McCleary} about completions). This together with exact sequence implies the isomorphism property of $e$.

\end{proof}

\section{ Properties of $\hat{\rE}$ as partially ordered set}\label{S:posetaf}
To proceed with our study of algebras $A_{\hdelta}^{\hdelta'}$ we need to set a terminology and to establish some simple properties of $\hat{\rE}$ that will be used later. This what we shall do presently.

The set of vertices of the graph $\hat{\rE}$ (\ref{P:kxccvgw13}) is a  union
\begin{equation}\label{E:disjoint}
\hat{\rE}=\coprod_{r\in \mathbb{Z}} \rE_r\quad \rE_r=\{\alpha^r|\alpha\in \rE\}
\end{equation} 
$\rE_r$   define full subgraphs in $\hat{\rE}$. 
We extend function $\rht$ (Definition \ref{D:rht}) to $\hat{\rE}$.
It is a matter of simple inspection  of (\ref{P:kxccvgw13}) to arrive at  the following  table of values of $\rht$:
\begin{equation}\label{E:htvalues}
\begin{array}{ll|c}
\hdelta'&\hdelta''&\rht\\\hline
(0)^r&(3)^{r-1}&8r\\
(12)^r&(2)^{r-1}&8r+1\\
(13)^r&(1)^{r-1}&8r+2\\
(14)^r&(23)^{r}&8r+3\\
(15)^r&(24)^{r}&8r+4\\
(25)^r&(34)^{r}&8r+5\\
(35)^r&(5)^{r}&8r+6\\
(45)^r&(4)^{r}&8r+7\\
\end{array}
\end{equation}

A shift operator $T$ acts on 
$\hat{\rE}$: 
\begin{equation}\label{E:shift}T \alpha^r=\alpha^{r+1}\quad \alpha^r\in \hat \rE
\end{equation}
It is an automorphism of the lattice $\hat{\rE}$ and $T^r\rE=\rE_r$.

The poset $\rE$ coincides with the interval $[(0),(1)]\subset \hat \rE$. 
Visual presentation of the subinterval $[(0),(1)^1]$ is given below:
\bigskip
\bigskip
\bigskip
\bigskip
\bigskip
\bigskip
\bigskip
\bigskip

\bigskip

{\tiny\begin{equation}\label{P:kxccvgw15}
\mbox{
\setlength{\unitlength}{3000sp}
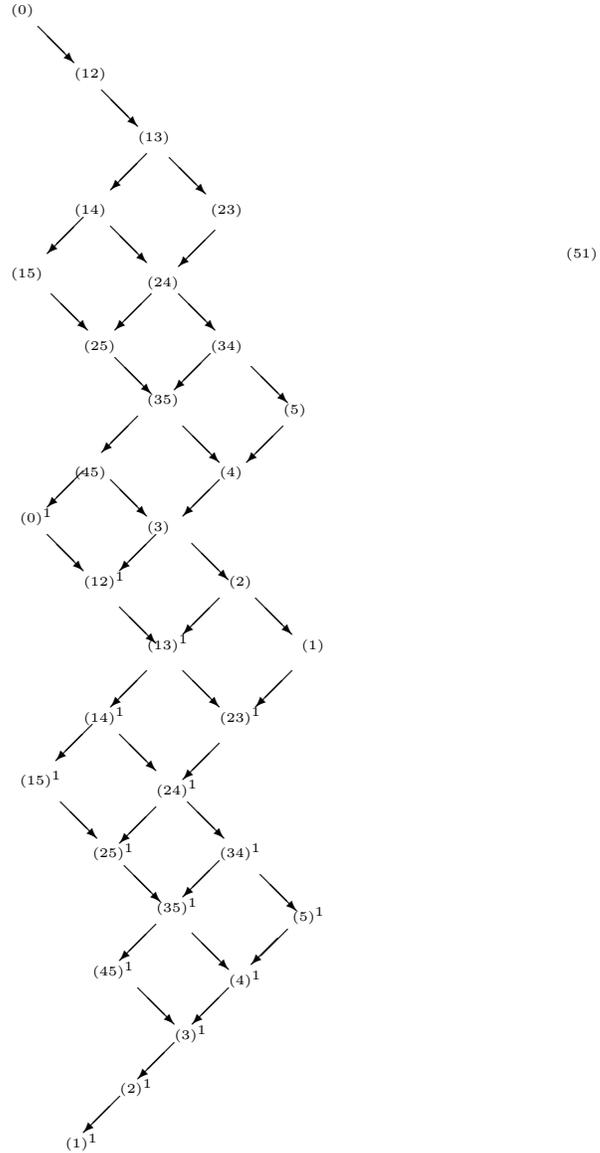
\begin{picture}(5025,-460)(-201,-1600)
\thinlines
\put(2326,-3886){\makebox(0,0)[lb]{\tiny{$(3)^{}$}}}
\put(2888,-1373){\vector(-1,-1){300}}
\put(2326,-736){\vector(-1,-1){300}}
\put(1801,-1261){\vector(-1,-1){300}}
\put(2025,-1335){\vector( 1,-1){300}}
\put(1537,-1897){\vector( 1,-1){300}}
\put(2063,-2423){\vector( 1,-1){300}}
\put(2851,-2386){\vector(-1,-1){300}}
\put(2363,-1898){\vector(-1,-1){300}}
\put(1951,-211){\vector( 1,-1){300}}
\put(2513,-773){\vector( 1,-1){300}}
\put(2588,-1898){\vector( 1,-1){300}}
\put(3188,-2498){\vector( 1,-1){300}}
\put(2626,-2986){\vector( 1,-1){300}}
\put(2700,-3960){\vector( 1,-1){300}}
\put(1501,-3886){\vector( 1,-1){300}}
\put(1876,-5461){\vector(-1,-1){300}}
\put(1612,-6097){\vector( 1,-1){300}}
\put(2138,-6623){\vector( 1,-1){300}}
\put(2926,-6586){\vector(-1,-1){300}}
\put(3488,-7148){\vector(-1,-1){300}}
\put(2663,-6098){\vector( 1,-1){300}}
\put(3263,-6698){\vector( 1,-1){300}}
\put(2701,-7186){\vector( 1,-1){300}}
\put(2401,-6136){\vector(-1,-1){300}}
\put(3226,-4411){\vector( 1,-1){300}}
\put(2926,-5611){\vector(-1,-1){300}}
\put(3526,-5011){\vector(-1,-1){300}}
\put(3451,-2986){\vector(-1,-1){300}}
\put(2101,-8536){\vector(-1,-1){300}}
\put(2401,-7111){\vector(-1,-1){300}}
\put(2926,-4411){\vector(-1,-1){300}}
\put(1801,-3361){\vector(-1,-1){300}}
\put(2926,-3436){\vector(-1,-1){300}}
\put(2251,-2911){\vector(-1,-1){300}}
\put(2026,-3436){\vector( 1,-1){300}}
\put(2401,-3886){\vector(-1,-1){300}}
\put(2626,-5011){\vector( 1,-1){300}}
\put(2101,-5536){\vector( 1,-1){300}}
\put(2251,-7636){\vector( 1,-1){300}}
\put(2551,-8086){\vector(-1,-1){300}}
\put(3001,-7636){\vector(-1,-1){300}}
\put(2326,-5011){\vector(-1,-1){300}}
\put(2101,-4486){\vector( 1,-1){300}}
\put(1726,-136){\makebox(0,0)[lb]{\tiny{$(12)^{}$}}}
\put(1201,389){\makebox(0,0)[lb]{\tiny{$(0)^{}$}}}
\put(2326,-1861){\makebox(0,0)[lb]{\tiny{$(24)^{}$}}}
\put(1801,-2386){\makebox(0,0)[lb]{\tiny{$(25)^{}$}}}
\put(1726,-1261){\makebox(0,0)[lb]{\tiny{$(14)^{}$}}}
\put(3451,-2911){\makebox(0,0)[lb]{\tiny{$(5)^{}$}}}
\put(2851,-2386){\makebox(0,0)[lb]{\tiny{$(34)^{}$}}}
\put(2851,-1261){\makebox(0,0)[lb]{\tiny{$(23)^{}$}}}
\put(2251,-661){\makebox(0,0)[lb]{\tiny{$(13)^{}$}}}
\put(1201,-1786){\makebox(0,0)[lb]{\tiny{$(15)^{}$}}}
\put(2926,-3436){\makebox(0,0)[lb]{\tiny{$(4)^{}$}}}
\put(2326,-2836){\makebox(0,0)[lb]{\tiny{$(35)^{}$}}}
\put(3001,-4336){\makebox(0,0)[lb]{\tiny{$(2)^{}$}}}
\put(1276,-3811){\makebox(0,0)[lb]{\tiny{$(0)^{1}$}}}
\put(1801,-4336){\makebox(0,0)[lb]{\tiny{$(12)^{1}$}}}
\put(2326,-4861){\makebox(0,0)[lb]{\tiny{$(13)^{1}$}}}
\put(2926,-5461){\makebox(0,0)[lb]{\tiny{$(23)^{1}$}}}
\put(1801,-5461){\makebox(0,0)[lb]{\tiny{$(14)^{1}$}}}
\put(2401,-6061){\makebox(0,0)[lb]{\tiny{$(24)^{1}$}}}
\put(2926,-6586){\makebox(0,0)[lb]{\tiny{$(34)^{1}$}}}
\put(3526,-7111){\makebox(0,0)[lb]{\tiny{$(5)^{1}$}}}
\put(2401,-7036){\makebox(0,0)[lb]{\tiny{$(35)^{1}$}}}
\put(1276,-5986){\makebox(0,0)[lb]{\tiny{$(15)^{1}$}}}
\put(1876,-6586){\makebox(0,0)[lb]{\tiny{$(25)^{1}$}}}
\put(2551,-8086){\makebox(0,0)[lb]{\tiny{$(3)^{1}$}}}
\put(3001,-7636){\makebox(0,0)[lb]{\tiny{$(4)^{1}$}}}
\put(3601,-4861){\makebox(0,0)[lb]{\tiny{$(1)^{}$}}}
\put(1651,-8986){\makebox(0,0)[lb]{\tiny{$(1)^{1}$}}}
\put(2101,-8536){\makebox(0,0)[lb]{\tiny{$(2)^{1}$}}}
\put(1876,-7561){\makebox(0,0)[lb]{\tiny{$(45)^{1}$}}}
\put(1726,-3436){\makebox(0,0)[lb]{\tiny{$(45)^{}$}}}
\put(1426,314){\vector( 1,-1){300}}
\end{picture}
}
\end{equation}} 
\bigskip
\bigskip
\bigskip
\bigskip
\bigskip
\bigskip
\bigskip
\bigskip
\bigskip
\bigskip
\bigskip
\bigskip
\bigskip
\bigskip
\bigskip
\bigskip
\bigskip
\bigskip
\bigskip
\bigskip
\bigskip
\bigskip
\bigskip
\bigskip
\bigskip
\bigskip
\bigskip
\bigskip
\bigskip

Let $A,B$ be two subsets of $\hat{\rE}$. We write $A\leq B$ if $a\leq b$ for all $a\in A$, $b\in B$. 

\begin{lemma}\label{E:dom}
The sets $\rE_r$ satisfy \[\rE_{r+k}\geq \rE_r, \rE_{r-k}\leq \rE_r \text{ {\rm for} }k\geq 2\]
\end{lemma}
\begin{proof}
The reader can convince himself by examining  the diagram (\ref{P:kxccvgw13}).
\end{proof}

As a corollary we obtain that if $(\hat{\alpha},\hat{\beta})$ is a clutter  then  $(\hat{\alpha},\hat{\beta})$ or  $(\hat{\alpha},T\hat{\beta})$ or  $(\hat{\alpha},T^{-1}\hat{\beta})$ are in $ \rE_r$.

\begin{proposition}
Poset $\hat \rE$ and all its subintervals  are latices.
\end{proposition}
\begin{proof}
It suffices to verify the statement only for $\hat \rE$.
We need to check that 
lower bound $\inf [\hat \alpha,\infty)\cup[\hat\beta,\infty)$  consists of one element; the same should  hold for  $\sup(-\infty,\hat \alpha]\cup(-\infty,\hat\beta]$. 
We know (Lemma \ref{E:dom}) that if $\hat \alpha \in \rE_{r+k}$ and $\hat \beta\in \rE_r , |k|\geq 2$ are comparable, $\#\inf{[\hat \alpha,\infty)\cup[\hat\beta,\infty)}=1$ and  $\hat \alpha\vee\hat \beta$ is defined. This argument also works for $\hat \alpha\wedge\hat \beta$.

By virtue of Lemma \ref{E:dom} and automorphism $T$ we may assume without loss of generality  that a pair of non comparable  elements $\hat \alpha,\hat \beta$ belong to $[(0),(1)^1]$, for which lattice property can be verified by direct inspection (see Picture (\ref{P:kxccvgw15})). 
\end{proof}
\subsection{Description of the affine Weyl group $W(\hat{D}_5)$}
We shall occupy ourself in this section with  a description of the affine Weyl group $W(\hat{D}_5)$ as a semidirect product and through generators and relations. 

 The group $W(\hat{D}_5)$ is a semidirect product $W(D_5)\ltimes X_*$ (cf. \cite{Humphreys}), where $X_*$ is a coroot lattice  of $\SO(10)$. The coroot lattice $X_*$ $\SO(10)$ coincides with the kernel of the homomorphism $\Z^5\rightarrow \Z_2$ defined by the formula \[\mm=[m_1,\dots,m_5]\rightarrow \sum_{i=1}^{5}m_i\mod 2\]This homomorphism is nothing else but a homomorphism  of fundamental groups $\pi_1(\T^5)\rightarrow \pi_1(\SO(10))$ induced by the inclusion of the maximal torus.

Semidirect product $S_5\ltimes (\Z_2)^5$ acts on $\Z^5$. The $ (\Z_2)^5$-factor act coordinatewise on $\mm=[m_1,\dots,m_5]\in \Z^5$ by changing  signs. Symmetric group shuffles components of $\mm$.

It is easy to see that  the product $s_{6}=\sigma_{4,5}(0,0,0,1,1)[0,0,0,1,1]\in S_5NX_*=W(\hat{D}_5)$ is an involution. The involutive  generators of $W(\hat{D}_5)$ can be arranged into the following Coxeter diagram, which makes $W(\hat{D}_5)$ a Coxeter group:

\[
 \renewcommand{\arraycolsep}{0.1cm}
 \renewcommand{\arraystretch}{1.4}
 \begin{matrix}
 s_{1} &\myline& s_{3} &\myline& s_{4} &\myline& s_{5} \\
&&\big|&&\big|  \\
 && \kern-1cm s_{2} \kern-1cm&&\kern-1cm s_{6}  \kern-1cm
 \end{matrix}
 \]
 
 We define the action of $W(\hat{D}_5)$ on the set 
 \begin{equation}\label{E:affineweights}
 \widehat{N}=N\times \Z
 \end{equation}
  by the formula:
 \begin{equation}\label{E:actionaffine}
 (\sigma,\epsilon,\mm)(\eta,n)=(\sigma(\epsilon+\eta),-\frac{1}{2}\sum_{i=1}^5(-1)^{\eta_i}m_i+n)
 \end{equation}
 The group $W(\hat{D}_5)$ acts transitively on $\widehat{N}$.
 
\begin{proposition}
There is an isomorphism of non-directed graphs \[Q(\widehat{N}, s_1,\dots,s_6)\cong \hat{\rE}.\]
\end{proposition}
\begin{proof}
If we utilize only $s_1,\dots,s_5$  we would get a disjoint union (\ref{E:disjoint}) (cf. discussion in Section \ref{S:finite}). 
Note that on $(14)^r,(24)^r,(34)^r,(5)^r$ the actions of $s_6$ and $s_5$ coincide. We conclude from formula (\ref{E:actionaffine}) that $s_6(45)^r=(0)^{r+1}, s_6(3)^r=(12)^{r+1},s_6(2)^r=(13)^{r+1},s_6(1)^r=(23)^{r+1}$. This verifies the claim.
\end{proof}
\subsection{Interpretation of $\hat{\rE}$ in terms of the spinor representation $S[z,z^-1]$}
The Lie-theoretic construction of $\hat{\rE}$ parallels the one of $\rE$.
We extend    operators $R^+_1,\dots,R^+_5$ (\ref{E:R1},\ref{E:R25})   to $S[z,z^{-1}]$ by $\mathbb{C}[z,z^{-1}]$-linearity.  Together with 
\[R^+_6=v_{4} v_{5}z \]  $R^+_1,\dots,R^+_5$ correspond to  simple root vectors  in the affine  $\hat{\so}_{10}$.
\begin{proposition}
We identify vertices of the graph $\hat{\rE}$ with weights of representation $S[z,z^{-1}]$. The operators $R^+_0,\dots,R^+_5$ intertwine the weight spaces. The corresponding Hasse diagram  coincides with $\hat{\rE}$.
\end{proposition}
\begin{proof}
By construction, the Hasse diagram of $S[z,z^{-1}]$ is a disjoint union (\ref{E:disjoint}) with $\rE_r$ identified with the diagram of $S\otimes z^r$. It remains to make use of $R^+_6$, which is left as an exercise for the reader.
\end{proof}

In our applications we shall need to know the structure of the set of  clutters $\hat{K}$ in $\hat{\rE}$.
The set $\hat{K}$ is invariant with respect to translations. This observation enables us to translate any clutter into 
$[(0),(1)^1]$ and carry its analysis on the case-by-case basis. We conclude that $\hat{K}$ besides shifts of clutters  from $M$ (\ref{E:Enoncomp}) contains shifts of 
\begin{equation}\label{E:noncomparableloop}
\begin{split}
&\{((0)^1,(5)),    ((0)^1,(4)),((0)^1,(3)),((0)^1,(2)),((0)^1,(1)), \\
&((12)^1,(2)),((12)^1,(1)),    ((13)^1,(1)),((14)^1,(1)),((15)^1,(1))\}
\end{split}
\end{equation}
\begin{lemma}\label{L:orbithatm}
Clutters in $\hat{\rE}$ belong to one $W(\hat{D}_5)$ orbit in $\Sym^2\hat{\rE}$.
\end{lemma}
\begin{proof}
Let $e_i$ be $[0,\dots,0,\overset{i}{1},0,\dots,0]\in X_*$. Denote $e_{i}+e_j,i\neq j$ by $\mm_{ij}$. We can conclude  from formula (\ref{E:actionaffine}) that $\mm_{ij}((0)^1,(i))=((0), (i))$ $(-\mm_{ij})((ij)^1,(j))=((ij), (j))$. We finish the proof as in Proposition \ref{E:transitiveM}.
\end{proof}
\section{Proof of the straightened law}\label{S:alghat}
This is a principal section of the present paper.  Recall that a proof of straightened law should consists of two parts. The easy part establishes the form (\ref{E:strrelgeneral}) of relations. The hard part ensures that standard monomials define a basis. We start with the easy part.

\begin{proposition}\label{P:rela}
Each quadric $\Gamma^{\hat{\s}}$ in the defining  relations 
 of algebras $A_{\hdelta}^{\hdelta'}$ (Definition \ref{D:algAdelta}) contains a unique clutter.
More over, for any clutter  $\lambda^{\halpha}\lambda^{\halpha'}$ straightened relation (\ref{E:reluniv}) holds.

\end{proposition}
\begin{proof}
We deduce statements of the proposition from the analogous statements for algebra $\hat A$. We follow closely the arguments of Proposition \ref{E:degreenull}.

A choice of complementary pair of components of gamma-maps $\Gamma^{+},\Gamma^{-}$ (as in the proof of Proposition \ref{E:degreenull}) and a subordinated choice of  $\so_8$, which leaves $\Gamma^{+},\Gamma^{-}$ invariant, allows us to construct a set of relations $\Gamma^{+,r},\Gamma^{-,r},r\in \mathbb{Z}$. By virtue of (\ref{E:spinorsplit}) we have an $\hat{\so_8}$-equivariant identification.
\begin{equation}\label{E:block}
S[z,z^{-1}]=S^+[z,z^{-1}]\times S^-[z,z^{-1}].
\end{equation} 
Let 
\[p_{\pm}:S[z,z^{-1}]\rightarrow S^{\pm}[z,z^{-1}]\]
be the projections. From results of Proposition  \ref{E:degreenull}    $\Gamma^{+,r}=p^*_{+}\Gamma^{+,r}|_{S^{+}[z,z^{-1}]}, \Gamma^{-,r}=p^*_{-}\Gamma^{-,r}|_{S^{-}[z,z^{-1}]}$. The set of quadrics $\Gamma^{+,r}|_{S^{+}[z,z^{-1}]}$ is an affinization of the quadric  $\Gamma^{+}|_{S^{+}}$.  This situation  has been studied in (\cite{Movq}). 
$S^+$ is equivalent by triality to the fundamental representation $V^8$ of $\so_{8}$. As in Proposition \ref{E:degreenull}, we conclude that Hasse diagrams of $V^8[z,z^{-1}]$ and $S^+[z,z^{-1}]$ coincide. We use the fact that the  lowest root in the adjoint representation of $\so_8$ is invariant with respect to our triality automorphism. This immediately follows from the cross shape of the extended Dynkin diagram corresponding to $\hat \so_8$.   The Hasse diagrams $\hat{\rG}_4$ of $V^8[z,z^{-1}]$

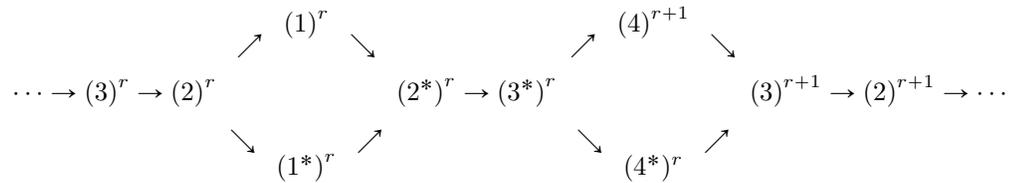
\[\begin{split}
&\cdots  \rightarrow  {(3)}^{r} \rightarrow  {(2)}^{r} \begin{array}{c}\nearrow  \begin{array}{c} {(1)}^{r} \\ \text{ } \end{array} \searrow \\ \searrow   \begin{array}{c} \text{ }\\   {(1^*)}^{r}  \end{array} \nearrow \end{array} {(2^*)}^{r}  \rightarrow    
 {(3^*)}^{r} \begin{array}{c}\nearrow  \begin{array}{c} {(4)}^{r+1}  \\ \text{ } \end{array} \searrow \\ \searrow   \begin{array}{c} \text{ }\\ (4^*)^{r}  \end{array} \nearrow \end{array} {(3)}^{r+1} \rightarrow {(2)}^{r+1} \rightarrow\cdots 
\end{split}
 \]
 has been determined in \cite{Movq}.

By the results of  \cite{Movq} each equation of  affinization of a nondegenerate quadric contains precisely  one clutter and relation (\ref{E:reluniv}) holds.
Arguing as in the proof of Proposition \ref{E:degreenull} and using result of Lemma \ref{L:orbithatm} we conclude that indices of the clutter remains to be incomparable in $\hat{\rE}$.
The poset associated with $\hat{\rG}_4$ is subposet in $\hat{\rE}$. Thus $\Gamma^{+,r}|_{S^{+}[z,z^{-1}]}$  lead to straightened relation  understood in the sense of  $\hat{\rE}$ order.

Let $q$ be the projection  $\hat P\rightarrow P_{\hdelta}^{\hdelta'}$. Suppose $\lambda^{\hgamma}\lambda^{\hgamma'}$ is a monomial in $\Gamma^{\s}$ such that $q(\lambda^{\hgamma}\lambda^{\hgamma'})\neq0$. Then $\hgamma,\hgamma'\in [\delta,\delta']$. By the above results, if $\hgamma,\hgamma'$ is not a clutter itself, the indices of the clutter $\lambda^{\halpha}\lambda^{\halpha'}$ in $\Gamma^{\s}$ satisfy $\hgamma\leq \halpha,\halpha'\leq \hgamma'$, which implies that $\halpha, \halpha'\in [\delta,\delta']$ and $q(\lambda^{\halpha}\lambda^{\halpha'})\neq 0$. We conclude that $q(\Gamma^{\s})$ contains a unique clutter. The claim follows because all relations in $A_{\hdelta}^{\hdelta'}$ have this form.

\end{proof}

The difficult part in the proof of straitened law claims that standard monomials are linearly independent. To prove this, we write relations (\ref{E:relgenfunction}) more explicitly:

{\tiny
\begin{equation}\label{E:equationsafexp}
\begin{split}
&\Gamma^{(1)^{2n}}= \underbracket{w^{n}_{2, 5} w^{n}_{3, 4}} - \overbracket{w^{n}_{2, 4} w^{n}_{3, 5}} +\cdots\sim 0 \quad
\Gamma^{(2)^{2n}}=-\underbracket{w^{n}_{1,5} w^{n}_{3,4}}+\overbracket{w^{n}_{1,4} w^{n}_{3,5}}+\cdots\sim 0\\
&\Gamma^{(3)^{2n}}=\underbracket{w^{n}_{1,5} w^{n}_{2,4}}-\overbracket{w^{n}_{1,4} w^{n}_{2,5}}+\cdots\sim 0 \quad
\Gamma^{(4)^{2n}}=-\underbracket{w^{n}_{1,5} w^{n}_{2,3}}+\overbracket{w^{n}_{1,3} w^{n}_{2,5}}+\cdots\sim 0\\
&\Gamma^{(5)^{2n}}=\underbracket{w^{n}_{1,4} w^{n}_{2,3}}-\overbracket{w^{n}_{1,3} w^{n}_{2,4}}+\cdots\sim 0 \quad
\Gamma^{(1^*)^{2n}}= - \overbracket{p^{n}_{4} w^{n}_{1, 4}} + \underbracket{p^{n}_{5} w^{n}_{1, 5}}+\cdots\sim 0\\
&\Gamma^{(2^*)^{2n}}= - \overbracket{p^{n}_{4} w^{n}_{2, 4}} + \underbracket{p^{n}_{5} w^{n}_{2, 5}}+\cdots\sim 0 \quad
\Gamma^{(3^*)^{2n}}= - \overbracket{p^{n}_{4} w^{n}_{3, 4}} + \underbracket{p^{n}_{5} w^{n}_{3, 5}}+\cdots\sim 0\\
&\Gamma^{(4^*)^{2n}}= - \overbracket{p^{n}_{3} w^{n}_{3, 4}} +\underbracket{p^{n}_{5} w^{n}_{4, 5}}+\cdots\sim 0 \quad
\Gamma^{(5^*)^{2n}}= - \overbracket{p^{n}_{3} w^{n}_{3, 5}} + \underbracket{p^{n}_{4} w^{n}_{4, 5}}+\cdots\sim 0\\
&\Gamma^{(1)^{2n+1}}=\underbracket{\lambda^{n+1} p^{n}_{1}} + \overbracket{w^{n+1}_{2, 3} w^{n}_{4, 5}}+\cdots\sim 0 \quad
\Gamma^{(2)^{2n+1}}=-\underbracket{\lambda^{n+1}  p^{n}_{2}}-\overbracket{w^{n+1}_{1,3} w^{n}_{4,5}}+\cdots\sim 0\\
&\Gamma^{(3)^{2n+1}}=\underbracket{\lambda^{n+1}  p^{n}_{3}}+\overbracket{w^{n+1}_{1,2} w^{n}_{4,5}}+\cdots\sim 0 \quad
\Gamma^{(4)^{2n+1}}=-\underbracket{\lambda^{n+1}  p^{n}_{4}}-\overbracket{w^{n+1}_{1,2} w^{n}_{3,5}}+\cdots\sim 0\\
&\Gamma^{(5)^{2n+1}}=\underbracket{\lambda^{n+1}  p^{n}_{5}}+\overbracket{w^{n+1}_{1,2} w^{n}_{3,4}}+\cdots\sim 0 \quad
\Gamma^{(1^*)^{2n+1}}=-\underbracket{p^{n}_{2} w^{n+1}_{1, 2}} + \overbracket{p^{n}_{3} w^{n+1}_{1, 3}} +\cdots\sim 0\\
&\Gamma^{(2^*)^{2n+1}}=-\underbracket{p^{n}_{1} w^{n+1}_{1, 2}} + \overbracket{p^{n}_{3} w^{n+1}_{2, 3}} +\cdots\sim 0 \quad
\Gamma^{(3^*)^{2n+1}}=-\underbracket{p^{n}_{1} w^{n+1}_{1, 3}} + \overbracket{p^{n}_{2} w^{n+1}_{2, 3}} +\cdots\sim 0\\
&\Gamma^{(4^*)^{2n+1}}=-\underbracket{p^{n}_{1} w^{n+1}_{1, 4}} + \overbracket{p^{n}_{2} w^{n+1}_{2, 4}} +\cdots\sim 0 \quad
\Gamma^{(5^*)^{2n+1}}=-\underbracket{p^{n}_{1} w^{n+1}_{1, 5}} + \overbracket{p^{n}_{2} w^{n+1}_{2, 5}} +\cdots\sim 0
\end{split}
\end{equation}
}
Monomials marked by underbracket $\underbracket{\ \ }$ are formed by incomparable variables. Monomials  $\lambda^{\alpha\wedge\beta}\lambda^{\alpha\vee\beta}$ are marked by overbracket $\overbracket{\ \ }$.

Construction   a basis $B_{\hdelta\ k}^{\hdelta'}$  in $k$-th graded component of $A_{\hdelta}^{\hdelta'}$ follows the line of a similar construction for algebra $A$. 
Monomials $\hat{\lambda}^{\bm{n}}=\prod_{\halpha \in [\hdelta, \hdelta']}(\lambda^{\hat\alpha})^{n_{\hat{\alpha}}}, \deg \lambda^{\bm{n}}=\overline{\bm{n}}=k$ form a basis in $P_{\hdelta\ k}^{\hdelta'}\subset P_{\hdelta}^{\hdelta'}$. 
We construct $B_{\hdelta\ k}^{\hdelta'}\subset\{\hat{\lambda}^{\bm{n}}|\overline{\bm{n}}=k\}$ and then project into $A_{\hdelta\ k}^{\hdelta'}$. We do it  in  several steps. 
The first approximation to $B_{\hdelta\ k}^{\hdelta'}$ is a subset $X_{\hdelta\ k}^{\hdelta'} \subset \{\hat{\lambda}^{\bm{n}}|\overline{\bm{n}}=k\}$ of standard  monomials $\hat\lambda^{\bm{n}}$.
 Comparable indices define a path in $[\hdelta, \hdelta']$, with weights $n_{\hat{\alpha}},\sum n_{\hat{\alpha}}=k $ at the vertices. The general theory tells us that there could be  sequence of sets  \[X_{\hdelta\ k}^{\hdelta'}\supseteq Y_{\hdelta\ k}^{\hdelta'}\supseteq \cdots \supseteq B_{\hdelta\ k}^{\hdelta'}\]
that converges to a basis. We shall prove however that $X_{\hdelta\ k}^{\hdelta'}=B_{\hdelta\ k}^{\hdelta'}$.
We define $\hat{X}_{k}=X^{\infty}_{-\infty\ k}$.

Let $v_{\bm{s}^r}=v_{\bm{s}}z^r, \bm{s}^r\in \hat{\rG}$ be  a basis $V[z,z^{-1}]$  build upon basis (\ref{E:basisvec}), and let $(x_{(\s)^l})=(x_{(1)^l},\dots,x_{(5)^l},x_{(1^*)^l},\dots,x_{(5^*)^l})\quad l\in \mathbb{Z}$ be the dual basis.
We note that Fierz identity (\ref{E:Fierz}) implies 
\begin{equation}\label{E:Fierzaf}
I_{\alpha}(z)=\Gamma^i_{\alpha\beta}\lambda^{\beta}(z)\Gamma^i_{\gamma\delta}\lambda^{\gamma}(z)\lambda^{\delta}(z)=0
\end{equation}
We define multiplication map 
 by the formula
\begin{equation}\label{E:subsaff}
\lambda^{\alpha^{l}}x_{\bm{s}^r}\rightarrow \lambda^{\alpha^{l}}\Gamma^{\bm{s}^r}
\end{equation}

Expressions 
\begin{equation}\label{E:Fierzafmode} 
h_{\alpha^k}=\sum_{l+l'=k}\sum_{\beta\in \rE}\sum_{i=1}^5\Gamma^{i}_{\alpha\beta}\lambda^{\beta^{l'}}x_{(i^*)^l}+\Gamma^{i^*}_{\alpha\beta}\lambda^{\beta^{l'}}x_{(i)^l}
\end{equation}
upon substitution (\ref{E:subsaff}) by virtue (\ref{E:Fierzaf}) become zero.

The action of operator  $T: h_{\alpha^k}\rightarrow  h_{\alpha^{k+3}}$  (\ref{E:shift}) is compatible with its action $T:\Gamma^{\bm{s}^l}\rightarrow \Gamma^{\bm{s}^{l+2}}$ on quadrics  (\ref{E:relgenfunction}) and on generators $\lambda^{\alpha^l}$. The operator $\hat{u}$ (\ref{E:inv}) transforms Fierz identities as follows:\[\hat{u}: h_{\alpha^k}\rightarrow  \pm h_{s(\alpha)^{-k}}.\]
In light of these remarks  the  following lemma is obvious.
\begin{lemma}\label{L:actionts}
Under the group generated by $T, \hat{u}$ any of the expressions (\ref{E:Fierzafmode} ) can be transformed to $h_{\alpha^0}$ or $h_{\alpha^1},\alpha\in \rE$.
\end{lemma}

Proposition \ref{P:rela} characterizes quadrics $\Gamma^{\hat{\bm{s}}}\in \hat{P}$ that map to zero in $P_{\hdelta}^{\hdelta'}$ as  those that contains  clutters with indices not  in $[\hdelta, \hdelta']$.
Recall that $\hat{P}$ and its graded components are spaces with topology. By $\hat{\otimes}$ we understand the tensor product completed in the topology.

Monomials $\lambda^{\hat{\gamma}} \otimes \lambda^{\hat{\alpha}} \lambda^{\hat{\beta}}$ with unconstrained indices form a topological  basis in $\hat{P}_1\hat{\otimes} \hat{P}_2$.

The map 
\begin{equation}\label{E:substaffone}
x_{\hat{\bm{s}}}\rightarrow \Gamma^{\hat{\bm{s}}}
\end{equation} identifies the second graded component $\hI_2$ of the ideal $\hI$  with the space dual to $V[z,z^{-1}]$. This allows us to think of $h_{\halpha}$ as elements in $ \hat{P}_1\hat{\otimes} \hI_2\subset \hat{P} \hat{\otimes} \hat{I}$.
\begin{proposition}\label{P:obstr}
An element $h_{\alpha^n}\in  \hat{P} \hat{\otimes} \hat{I}$ contains precisely a pair of noncommutative obstructions  $ \lambda^{\hat{\gamma}} \otimes  \lambda^{\hat{\alpha}} \lambda^{\hat{\beta}}$,  $\lambda^{\hat{\gamma}'}  \otimes \lambda^{\hat{\alpha}'} \lambda^{\hat{\beta}'}$ (see Definition \ref{D:obstruction}), such that their images coincide in $\hat{P}$. Conversely, for any pair of noncommutative obstructions $\lambda^{\hat{\gamma}} \otimes  \lambda^{\hat{\alpha}} \lambda^{\hat{\beta}} \neq \lambda^{\hat{\gamma}'}  \otimes \lambda^{\hat{\alpha}'} \lambda^{\hat{\beta}'}$ with $\lambda^{\hat{\gamma}}\lambda^{\hat{\alpha}} \lambda^{\hat{\beta}} =\lambda^{\hat{\gamma}'}  \lambda^{\hat{\alpha}'} \lambda^{\hat{\beta}'}$ there is $h_{\alpha^n}\in \hat{P} \hat{\otimes} \hat{I}$ that contains these obstructions as monomials.
\end{proposition}
\begin{proof}
The transformations $T,\hat{u}$ act simultaneously on the ideal $\hI$ and on the poset $\hat{\rE}$ and induce transformations on the sets of clutters and obstructions. By the results of Lemma \ref{L:actionts}, it suffices to analyze $h_{\alpha^0}$ and $h_{\alpha^1}$. 

The content of $h_{\alpha^1}$ is tabulated below.
\begin{equation}\label{E:fierzafobstr}
{\tiny
\begin{split}
&h_{(0)^1}= p^{1}_{1} x_{(1^*)^{0}}-p^{1}_{2} x_{(2^*)^{0}}+p^{1}_{3} x_{(3^*)^{0}}-{p^{1}_{4} x_{(4^*)^{0}}}+{p^{1}_{5} x_{(5^*)^{0}}} +  \underbracket{p^{0}_{1} x_{(1^*)^{1}}}-\underbracket{p^{0}_{2} x_{(2^*)^{1}}}+p^{0}_{3} x_{(3^*)^{1}}-{p^{0}_{4} x_{(4^*)^{1}}}+{p^{0}_{5} x_{(5^*)^{1}}} \cdots=0 \\
&h_{(12)^1}={w^{1}_{4,5} x_{(3^*)^{0}}}-{w^{1}_{3,5} x_{(4^*)^{0}}}+w^{1}_{3,4} x_{(5^*)^{0}}-p^{1}_{2} x_{(1)^{0}}-p^{1}_{1} x_{(2)^{0}}+ {w^{0}_{4,5} x_{(3^*)^{1}}}-{w^{0}_{3,5} x_{(4^*)^{1}}}+w^{0}_{3,4} x_{(5^*)^{1}}-\underbracket{p^{0}_{2} x_{(1)^{1}}}-\underbracket{p^{0}_{1} x_{(2)^{1}}}\cdots=0\\
&h_{(13)^1}=-{w^{1}_{4, 5} x_{(2^*)^{0}}} + {w^{1}_{2, 5} x_{(4^*)^{0}}} - w^{1}_{2, 4} x_{(5^*)^{0}} + p^{1}_{3} x_{(1)^{0}} - p^{1}_{1} x_{(3)^{0}} -{w^{0}_{4, 5} x_{(2^*)^{1}}} + {w^{0}_{2, 5} x_{(4^*)^{1}}} - w^{0}_{2, 4} x_{(5^*)^{1}} + \underbracket{p^{0}_{3} x_{(1)^{1}}} - \underbracket{p^{0}_{1} x_{(3)^{1}}}\cdots=0\\
&h_{(23)^1}= {w^{1}_{4, 5} x_{(1^*)^{0}}} - {w^{1}_{1, 5} x_{(4^*)^{0}}} + w^{1}_{1, 4} x_{(5^*)^{0}} + p^{1}_{3} x_{(2)^{0}} + p^{1}_{2} x_{(3)^{0}}  +{w^{0}_{4, 5} x_{(1^*)^{1}}} - {w^{0}_{1, 5} x_{(4^*)^{1}}} + w^{0}_{1, 4} x_{(5^*)^{1}} + \underbracket{p^{0}_{3} x_{(2)^{1}}} + \underbracket{p^{0}_{2} x_{(3)^{1}}}\cdots=0\\
&h_{(14)^1}={w^{1}_{3, 5} x_{(2^*)^{0}}} - {w^{1}_{2, 5} x_{(3^*)^{0}}} + w^{1}_{2, 3} x_{(5^*)^{0}} - p^{1}_{4} x_{(1)^{0}} - p^{1}_{1} x_{(4)^{0}}+ {w^{0}_{3, 5} x_{(2^*)^{1}}} - {w^{0}_{2, 5} x_{(3^*)^{1}}} + w^{0}_{2, 3} x_{(5^*)^{1}} - \underbracket{p^{0}_{4} x_{(1)^{1}}} - \underbracket{p^{0}_{1} x_{(4)^{1}}}\cdots=0\\
&h_{(24)^1}=-{w^{1}_{3, 5} x_{(1^*)^{0}}} + {w^{1}_{1, 5} x_{(3^*)^{0}}} - w^{1}_{1, 3} x_{(5^*)^{0}} - p^{1}_{4} x_{(2)^{0}} + p^{1}_{2} x_{(4)^{0}} -{w^{0}_{3, 5} x_{(1^*)^{1}}} + {w^{0}_{1, 5} x_{(3^*)^{1}}} - w^{0}_{1, 3} x_{(5^*)^{1}} - \underbracket{p^{0}_{4} x_{(2)^{1}}} + \underbracket{p^{0}_{2} x_{(4)^{1}}}\cdots=0\\
&h_{(15)^1}=-{w^{1}_{3, 4} x_{(2^*)^{0}}} + w^{1}_{2, 4} x_{(3^*)^{0}} - w^{1}_{2, 3} x_{(4^*)^{0}} + {p^{1}_{5} x_{(1)^{0}}} - p^{1}_{1} x_{(5)^{0}} -{w^{0}_{3, 4} x_{(2^*)^{1}}} + w^{0}_{2, 4} x_{(3^*)^{1}} - w^{0}_{2, 3} x_{(4^*)^{1}} + \underbracket{p^{0}_{5} x_{(1)^{1}}} - \underbracket{p^{0}_{1} x_{(5)^{1}}}\cdots=0\\
&h_{(34)^1}={w^{1}_{2, 5} x_{(1^*)^{0}}} - {w^{1}_{1, 5} x_{(2^*)^{0}}} + w^{1}_{1, 2} x_{(5^*)^{0}} - p^{1}_{4} x_{(3)^{0}} - p^{1}_{3} x_{(4)^{0}} +{w^{0}_{2, 5} x_{(1^*)^{1}}} - {w^{0}_{1, 5} x_{(2^*)^{1}}} + w^{0}_{1, 2} x_{(5^*)^{1}} - \underbracket{p^{0}_{4} x_{(3)^{1}}} - \underbracket{p^{0}_{3} x_{(4)^{1}}}\cdots=0\\
&h_{(25)^1}={w^{1}_{3, 4} x_{(1^*)^{0}}} - w^{1}_{1, 4} x_{(3^*)^{0}} + w^{1}_{1, 3} x_{(4^*)^{0}} + {p^{1}_{5} x_{(2)^{0}}} + p^{1}_{2} x_{(5)^{0}} +{w^{0}_{3, 4} x_{(1^*)^{1}}} - w^{0}_{1, 4} x_{(3^*)^{1}} + w^{0}_{1, 3} x_{(4^*)^{1}} + \underbracket{p^{0}_{5} x_{(2)^{1}}} + \underbracket{p^{0}_{2} x_{(5)^{1}}}\cdots=0\\
&h_{(5)^1}=\underbracket{\lambda^{1}  x_{(5^*)^{0}}}+{w^{1}_{1,5} x_{(1)^{0}}}+{w^{1}_{2,5} x_{(2)^{0}}}+w^{1}_{3,5} x_{(3)^{0}}+w^{1}_{4,5} x_{(4)^{0}} +\lambda^{0}  x_{(5^*)^{1}}+{w^{0}_{1,5} x_{(1)^{1}}}+{w^{0}_{2,5} x_{(2)^{1}}}+w^{0}_{3,5} x_{(3)^{1}}+\underbracket{w^{0}_{4,5} x_{(4)^{1}}}\cdots=0\\
&h_{(35)^1}=-{w^{1}_{2, 4} x_{(1^*)^{0}}} + w^{1}_{1, 4} x_{(2^*)^{0}} - w^{1}_{1, 2} x_{(4^*)^{0}} + {p^{1}_{5} x_{(3)^{0}}} - p^{1}_{3} x_{(5)^{0}} -{w^{0}_{2, 4} x_{(1^*)^{1}}} + w^{0}_{1, 4} x_{(2^*)^{1}} - w^{0}_{1, 2} x_{(4^*)^{1}} + \underbracket{p^{0}_{5} x_{(3)^{1}}} - \underbracket{p^{0}_{3} x_{(5)^{1}}}\cdots=0\\
&h_{(4)^1}=-\underbracket{\lambda^{1} x_{(4^*)^{0}}} - w^{1}_{1, 4} x_{(1)^{0}} - {w^{1}_{2, 4} x_{(2)^{0}}} - {w^{1}_{3, 4} x_{(3)^{0}}} +  w^{1}_{4, 5} x_{(5)^{0}} -\lambda^{0} x_{(4^*)^{1}} - w^{0}_{1, 4} x_{(1)^{1}} - {w^{0}_{2, 4} x_{(2)^{1}}} - {w^{0}_{3, 4} x_{(3)^{1}}} +  \underbracket{w^{0}_{4, 5} x_{(5)^{1}}}\cdots=0\\
&h_{(45)^1}={w^{1}_{2, 3} x_{(1^*)^{0}}} - w^{1}_{1, 3} x_{(2^*)^{0}} + w^{1}_{1, 2} x_{(3^*)^{0}} + {p^{1}_{5} x_{(4)^{0}}} + p^{1}_{4} x_{(5)^{0}}+ {w^{0}_{2, 3} x_{(1^*)^{1}}} - w^{0}_{1, 3} x_{(2^*)^{1}} + w^{0}_{1, 2} x_{(3^*)^{1}} + \underbracket{p^{0}_{5} x_{(4)^{1}}} + \underbracket{p^{0}_{4} x_{(5)^{1}}}\cdots=0\\
&h_{(3)^1}=\underbracket{\lambda^{1}  x_{(3^*)^{0}}}+w^{1}_{1,3} x_{(1)^{0}}+{w^{1}_{2,3} x_{(2)^{0}}}-{w^{1}_{3,4} x_{(4)^{0}}}-w^{1}_{3,5} x_{(5)^{0}}+ \lambda^{0}  x_{(3^*)^{1}}+w^{0}_{1,3} x_{(1)^{1}}+{w^{0}_{2,3} x_{(2)^{1}}}-{w^{0}_{3,4} x_{(4)^{1}}}-\underbracket{w^{0}_{3,5} x_{(5)^{1}}}\cdots=0\\
&h_{(2)^1}=-\underbracket{\lambda^{1} x_{(2^*)^{0}}} - w^{1}_{1, 2} x_{(1)^{0}} + {w^{1}_{2, 3} x_{(3)^{0}}} + {w^{1}_{2, 4} x_{(4)^{0}}} +  w^{1}_{2, 5} x_{(5)^{0}} -\lambda^{0} x_{(2^*)^{1}} - w^{0}_{1, 2} x_{(1)^{1}} + {w^{0}_{2, 3} x_{(3)^{1}}} + {w^{0}_{2, 4} x_{(4)^{1}}} +  \underbracket{w^{0}_{2, 5} x_{(5)^{1}}}\cdots=0\\
&h_{(1)^1}=\underbracket{\lambda^{1} x_{(1^*)^{0}}} - w^{1}_{1, 2} x_{(2)^{0}} - w^{1}_{1, 3} x_{(3)^{0}} -{ w^{1}_{1, 4} x_{(4)^{0}}} - {w^{1}_{1, 5} x_{(5)^{0}}} +\lambda^{0} x_{(1^*)^{1}} - w^{0}_{1, 2} x_{(2)^{1}} - w^{0}_{1, 3} x_{(3)^{1}} -{ w^{0}_{1, 4} x_{(4)^{1}}} - \underbracket{w^{0}_{1, 5} x_{(5)^{1}}}\cdots=0
\end{split}
}
\end{equation}
The underlined terms after substitution (\ref{E:substaffone})  contain complementary pairs of noncommutative obstructions (\ref{E:ncobstraf}) .
The reader can check using equations (\ref{E:relgenfunction}) that noncommutative obstructions have coefficients equal to $\pm1$ and have opposite signs. There is another way to become convinced in the choice of coefficients of obstructing monomials:  the image of $h_{\hat{\alpha}}$ in $\hat{P}$ is zero and for this coefficients of obstructions must have opposite signs.
The omitted terms in (\ref{E:fierzafobstr}) do not contribute to the set of obstructions because of results of Lemma \ref{E:dom} and explicit description of clutters in $\Gamma^{\hat{\bm{s}}}$ (\ref{E:relgenfunction})

Zero modes  in $h_{\alpha^0}$ coincide with $h_{\alpha}$ (\ref{E:fierzksimple}). Analysis of the remaining modes in  $h_{\alpha^0}$ can be carried out along the lines of $h_{\alpha^1}$ case.

The converse statement of the proposition can be proved as follows. First, take an obstruction and transform it to an obstruction in $[(0),(1)^1]$ by means of  $T$ and $\hat{u}$.   Then, a case-by-case  study (Appendix \ref{A:obstr}) reveals that  obstruction in  $[(0),(1)^1]$ are in one-to-one correspondence  to a noncommutative obstructions in some $h_{\alpha^0}$, $h_{\alpha^1}$ or $h_{\alpha^2}$. The $h_{\alpha^2}$-case as we already know can be reduced to  $h_{\alpha^1}$-case by means of automorphisms $T$ and $\hat{u}$.
\end{proof}

\begin{proposition}\label{P:obstrfin}
For  any pair of complementary noncommutative obstruction $\lambda^{\halpha}\otimes \lambda^{\hbeta}\lambda^{\hgamma}, \lambda^{\halpha'}\otimes \lambda^{\hbeta'}\lambda^{\hgamma'}\in P_{\hdelta}^{\hdelta'}\otimes I_{\hdelta}^{\hdelta'}$ there is $h_{\halpha}$ (\ref{E:Fierzafmode}) such that after the substitution (\ref{E:substaffone}) and projection  $p:P\rightarrow P_{\hdelta}^{\hdelta'}$ expression  $h_{\halpha}$ transforms to an expression that contains  monomials $\lambda^{\halpha}\otimes \lambda^{\hbeta}\lambda^{\hgamma}, \lambda^{\halpha'}\otimes \lambda^{\hbeta'}\lambda^{\hgamma'}$.
\end{proposition}
\begin{proof}
The map $p$ acts as an identity on $\lambda^{\halpha},\halpha\in [\hdelta,\hdelta']$ and by zero on $\lambda^{\halpha},\halpha\notin [\hdelta,\hdelta']$. Only monomials $\prod_{i}\lambda^{\halpha_i}$ with $\halpha_i\in  [\hdelta,\hdelta']$ survive under the map $p$. 
A noncommutative obstruction  defines two $[\hdelta,\hdelta']$-clutters.    They are still clutters in $\hat{\rE}$. Hence $\lambda^{\halpha}\otimes \lambda^{\hbeta}\lambda^{\hgamma}, \lambda^{\halpha'}\otimes \lambda^{\hbeta'}\lambda^{\hgamma'}$ is  a complementary pair of noncommutative obstructions  in $\hat{P}\hat{\otimes} \hI$. Hence by Proposition \ref{P:obstr} there is a $h_{\halpha}\in \hat{P}\hat{\otimes} \hI$. Its image in $P_{\hdelta}^{\hdelta'}\otimes I_{\hdelta}^{\hdelta'}$  contains $\lambda^{\halpha}\otimes \lambda^{\hbeta}\lambda^{\hgamma}$. 
\end{proof}

We define a strict total order   on the set of generators of $\hat{P}$, which refines a partial order defined on $\hat{\rE}$. The operator $T$ preserves the order:
\begin{equation}\label{E:orderaf}\begin{split}&\dots <(4)^{r-1}< (0)^r <(3)^{r-1}<(12)^r<(2)^{r-1}<(13)^r<(1)^{r-1}<(14)^r<(23)^r<\\
&<(15)^r<(24)^r<(25)^r<(34)^r<(35)^r<(5)^r<(45)^r<(4)^r<(0)^{r+1}<\dots\end{split}\end{equation}
The total order on the set of generators of $P_{\hdelta}^{\hdelta'}$ is the restriction of the above. The set of relations  (\ref{E:order}) is compatible with (\ref{E:orderaf}).

\begin{proposition}\label{P:Buchberger1}
If for any pair   $\Gamma^{\hat{\bm{s}}},\Gamma^{\hat{\bm{s}}'}\in I_{\hdelta}^{\hdelta'}\subset P_{\hdelta}^{\hdelta'}$ (\ref{E:relgenfunction}) with $\rT(\Gamma^{\hat{\bm{s}}},\Gamma^{\hat{\bm{s}}'})\neq \rT(\Gamma^{\hat{\bm{s}}})\rT(\Gamma^{\hat{\bm{s}}'})$ the $S$-polynomials are  
\begin{equation}\label{E:bush}
S(\Gamma^{\hat{\bm{s}}},\Gamma^{\hat{\bm{s}}'})=\sum_{\hat{\alpha}\in \hat{\rE}}\sum_{\hat{r}\in \hat{\rG}} c_{\hat{\alpha},\hat{r}}\lambda^{\hat{\alpha}}\Gamma^{\hat{r}}
\end{equation} with no obstructions among monomials in $\lambda^{\hat{\alpha}}\Gamma^{\hat{r}}$ , then the set $X_{\hdelta\ k}^{\hdelta'}$ defines a basis in ${A}_{\hdelta\ k}^{\hdelta'}$.
\end{proposition}
\begin{proof}
The proof is bases on the result of Proposition \ref{P:obstrfin}. We omit it because it  is similar to the proof of  Proposition \ref{P:quadbasis}.
\end{proof}

\begin{corollary}\label{P:genstrlaw}
The algebra $A_{\hdelta}^{\hdelta'}$ is an algebra with straightened law.
\end{corollary}
\begin{proof}
The same as of Proposition \ref{P:quadbasis}, though based on 
 Propositions \ref{P:rela} and \ref{P:Buchberger1}.
\end{proof}

\begin{corollary}\label{P:quadbasiscaf} Let us fix  the order (\ref{E:order}) on generators of $P_{\hdelta}^{\hdelta'}$.
Then the relations (\ref{E:relgenfunction}) form a quadratic Gr\"{o}bner basis  in the ideal $I_{\hdelta}^{\hdelta'}$. In particular algebras $A_{\hdelta}^{\hdelta'}$ are Koszul.
\end{corollary}
\begin{proof}
By virtue of general theory  \cite{Mora} \cite{PP} result  follows from (\ref{P:Buchberger1}).
\end{proof}
\begin{proposition}\label{P:reduce}
\begin{enumerate}
\item $Spec(A_{\hdelta}^{\hdelta'})$ is a reduced affine scheme.
\item $\dim Spec(A_{\hdelta}^{\hdelta'})=\rht(\hdelta')-\rht(\hdelta)$.
\end{enumerate}
\end{proposition}
\begin{proof}
The proof follows from Corollary 3.6 in \cite{EisenbudStr}. Corollary reduces computation of dimension to finding the length of the maximal chain $\alpha_1<\cdots<\alpha_d$ in $[{\hdelta},{\hdelta'}]$. For $[{\hdelta},{\hdelta'}]\subset\hat{\rE}$ it is expressible in terms of $\rht$ (\ref{E:ht},\ref{E:htvalues}) and coincides with $\rht(\hdelta')-\rht(\hdelta)$. 
\end{proof}

Depth $\depth(A_{\hdelta}^{\hdelta'})$ is defined as a length of the maximal regular sequence  in the ideal in $A_{\hdelta}^{\hdelta'}$ generated by $\lambda^{\hat{\alpha}}$. 
\begin{proposition}\label{P:depth}
\[\depth(A_{\hdelta}^{\hdelta'})=\dim Spec(A_{\hdelta}^{\hdelta'})=\rht(\hdelta')-\rht(\hdelta)\]
We set  \[y_i=\sum_{\rht(\hat{\alpha})=i} \lambda^{\hat{\alpha}},\] 
with $\rht$ defined in (\ref{E:ht}).  Then $y_{\rht(\hdelta)},\dots,y_{\rht(\hdelta')}$ form a regular sequence in $A_{\hdelta}^{\hdelta'}$.
\end{proposition}
\begin{proof}
An element $\beta$ in $F$ is a {\it cover} of $\alpha$ in $F$ if $\beta>{\alpha}$ and no elements lies strictly in between ${\alpha}$ and ${\beta}$.

The poset $F$ is called {\it wonderful} \cite{EisenbudStr} if in the poset $H\cup \{\infty\}\cup\{-\infty\}$ obtained by adjoining greatest and the least elements (if they are not already present), the following condition holds: if $\beta_1,\beta_2<\gamma$ are covers of an element $\alpha$, then there is an element $\beta<\gamma$, which covers both $\beta_1$ and $\beta_2$.

An easy inspection shows that $\hat{\rE}$ and $[{\hdelta},{\hdelta'}]$ are wonderful. The proposition follows Theorem 4.1 in \cite{EisenbudStr}, which establish this in a greater generality for algebras with straightened law over a wonderful poset.
\end{proof}

\begin{corollary}\label{C:CohenMacaulay}
$A_{\hdelta}^{\hdelta'}$ is a Cohen-Macaulay algebra.
\end{corollary}
\begin{proof}
Follows from Corollary 4.2 \cite{EisenbudStr}.
\end{proof}

\section{Computation of Poincar\'{e} series $A_{\hdelta}^{\hdelta'}(z,q,t)$}\label{S:genfunaff}
Chain counting method that we developed for the purpose of computation of character of pure spinors is fairly general and relies on combinatorics of Hasse diagram. This is why we can carry over  most of the results of Section \ref{S:Poincare} to  setting of algebras $A_{\hdelta}^{\hdelta'}$.

To be precise by the result of Corollary \ref{P:genstrlaw}  standard monomials form a basis in $A_{\hdelta}^{\hdelta'}$. Standard monomials are in one-to-one correspondence with chains in $[\hdelta,{\hdelta'}]\subset \hat{\rE}$. 
A generating function $A_{\hdelta}^{\hdelta'}(z,q,t)$ that encodes Lie-algebraic information about $A_{\hdelta}^{\hdelta'}$ is a specialization of  $C([\hdelta,{\hdelta'}])(t)$ (\ref{E:Fgenfun}).

To be precise, graded components of $A_{\hdelta}^{\hdelta'}$ are representations of a maximal torus $H\subset \Spin(10)$; $\mathbb{C}^{\times}$ acts on coefficients of $\lambda^{\alpha}(z)$ (\ref{E:genfunl}) via reparametrization: \[\lambda^{\alpha}(z)\rightarrow \lambda^{\alpha}(qz).\] 
Let $a_n(z,q)$ be the character of the $ H\times \mathbb{C}^{\times}$ represented in $n$-th graded component of $A_{\hdelta}^{\hdelta'}$. We define \[A_{\hdelta}^{\hdelta'}(z,q,t)=\sum_{n\geq 0}a_n(z,q)t^n.\]
We conclude, as in Section \ref{S:Poincare}, that $C([\hdelta,{\hdelta'}])(t)$, upon substitution  
\begin{equation}\label{E:affsubst}
e_{\alpha^r}=q^re_{\alpha}(z_1,\dots,z_5),
\end{equation} becomes  $A_{\hdelta}^{\hdelta'}(t,z,q)$.

The isomorphism $T^N:A_{\hdelta}^{\hdelta'}\rightarrow A_{T^N\hdelta}^{T^N\hdelta'}$  shifts $\mathbb{C}^{\times}$-weights of all generators on $N$ and
\begin{equation}\label{E:chshift}
A_{\alpha^N}^{\beta^{N'}}(z,q,t)=A_{\alpha}^{\beta^{N'-N}}(z,q,q^{N}t)\quad \alpha,\beta\in \rE
\end{equation}

It is convenient to group generating functions in arrays \[\row{C([\hdelta,{\hdelta'}])(t)}{C([\hdelta,{\hdelta''}])(t)} \quad \hdelta'\neq \hdelta'' \quad \rht(\hdelta')= \rht(\hdelta'')=l\]
(cf. \ref{E:htvalues}).  The ordered pair $(\hdelta',\hdelta'')$ shall be called a {\it complementary} pair of weights. Lattice $\hat{\rE}$ is convenient and narrow (see Section \ref{S:Poincare}). The formulas (\ref{E:tail} \ref{E:reduction}) can be written in a matrix form
\begin{equation}\label{E:matrixreq}
\begin{split}
&\row{C([\hdelta,{\hdelta'}])(t)}{C([\hdelta,{\hdelta''}])(t)}=\row{C([\hdelta,{\hepsilon'}])(t)}{C([\hdelta,{\hepsilon''}])(t)}U_l  \\
& \rht(\hepsilon')+1= \rht(\hepsilon')+1=\rht(\hdelta')= \rht(\hdelta'')=l
\end{split}
\end{equation}
Upon substitution (\ref{E:affsubst}), $U_{l}(z,q,t)$ satisfies  \[U_{l+8}(z,q,t)=U_l(z,q,qt)\]
This enables to effectively recover all $U_{l}(z,q,t)$ from the first eight matrices:
\begin{equation}
\begin{split}
&U_{1}(z,q,t)=\mat{\frac{e_{(0)}t}{1-e_{(12)}t}}{0}{\frac{1}{1-e_{(12)}t}}{\frac{1}{1-e_{(2)}q^{-1}t}}\quad U_{2}(z,q,t)=\mat{\frac{1}{1-e_{(13)}t}}{0}{\frac{e_{(2)}q^{-1}t}{1-e_{(13)}t}}{\frac{1}{1-e_{(1)}q^{-1}t}}\\
&U_{3}(z,q,t)=\mat{\frac{1}{1-e_{(14)}t}}{\frac{1}{1-e_{(23)}t}}{0}{\frac{e_{(1)}q^{-1}t}{1-e_{(23)}t}}
\quad U_{4}(z,q,t)=\mat{\frac{1}{1-e_{(15)}t}}{\frac{e_{(14)}t}{1-e_{(24)}t}}{0}{\frac{1}{1-e_{(24)}t}}\\
&U_{5}(z,q,t)=\mat{\frac{e_{(15)}t}{1-e_{(25)}t}}{0}{\frac{1}{1-e_{(25)}t}}{\frac{1}{1-e_{(34)}t}}
\quad U_{6}(z,q,t)=\mat{\frac{1}{1-e_{(35)}t}}{0}{\frac{e_{(34)}t}{1-e_{(35)}t}}{\frac{1}{1-e_{(5)}t}}\\
&U_{7}(z,q,t)=\mat{\frac{1}{1-e_{(45)}t}}{\frac{1}{1-e_{(4)}t}}{0}{\frac{e_{(5)}t}{1-e_{(4)}t}} \quad U_{8}(z,q,t)=\mat{\frac{1}{1-e_{(0)}qt}}{\frac{e_{(45)}t}{1-e_{(3)}t}}{0}{\frac{1}{1-e_{(3)}t}}\\
\end{split}
\end{equation}
Let us write a formula for $A_{\halpha}^{\hdelta}(z,q,t)$. 

Suppose $\halpha,\hdelta$ are from the first column in the table (\ref{E:htvalues}).
We choose $\hdelta'$  to be complementary to $\hdelta$.
If we make a substitution (\ref{E:affsubst}) into  iterated relation (\ref{E:matrixreq}) we obtain
\[\begin{split}
&A_{\halpha}^{\hdelta}(z,q,t)=\row{A_{\halpha}^{\hdelta}(z,q,t)}{A_{\halpha}^{\hdelta'}(z,q,t)}\col{1}{0}=\\
&=\row{A_{\halpha}^{\halpha}(z,q,t)}{0}U_{\rht(\halpha)} \cdots U_{\rht(\hdelta)}\col{1}{0} 
\end{split}\]
We define a function 
\[k(\halpha)=
\begin{cases}
0&\mbox{ if } \halpha \mbox{ is in the first column of table (\ref{E:htvalues}) } \\
1& \mbox{ otherwise }
\end{cases}
\]
Suppose $\halpha=\alpha^r$. Here is a general formula for the character: 
\begin{equation}\label{E:genchar}\begin{split}
&A_{\halpha}^{\hdelta}(z,q,t)=\row{\frac{1}{1-e_{\alpha}q^rt}}{0}\mat{0}{1}{1}{0}^{k(\halpha)}U_{\rht(\halpha)} \cdots U_{\rht(\hdelta)}\mat{0}{1}{1}{0}^{k(\hdelta)}\col{1}{0} 
\end{split}\end{equation}
In this equation normalization formula (\ref{E:norm}) was taken into account.

In full analogy with equation (\ref{E:matrixreq}) we obtain a set of equations for varying  lower indices.
\begin{equation}
\begin{split}
&\row{A^{\hdelta}_{(0)}(t)}{A^{\hdelta}_{(3)^{-1}}(t)}=\row{A^{\hdelta}_{(12)}(t)}{A^{\hdelta}_{(2)^{-1}}(t)} 
\mat{\frac{1}{1-e_{(0)}t}}{\frac{e_{(12)}t}{1-e_{(3)}q^{-1}t}}{0}{\frac{1}{1-e_{(3)}q^{-1}t}}\quad l=1\\
&\row{A^{\hdelta}_{(12)}(t)}{A^{\hdelta}_{(2)^{-1}}(t)} =\row{A^{\hdelta}_{(13)}(t)}{A^{\hdelta}_{(1)^{-1}}(t)} 
\mat{\frac{1}{1-e_{(12)}t}}{\frac{1}{1-e_{(2)}q^{-1}t}}{0}{\frac{e_{(1)}q^{-1}t}{1-e_{(2)}t}}\quad l=2\\
&\row{A^{\hdelta}_{(13)}(t)}{A^{\hdelta}_{(1)^{-1}}(t)} =\row{A^{\hdelta}_{(14)}(t)}{A^{\hdelta}_{(23)}(t)} 
\mat{\frac{1}{1-e_{(13)}t}}{0}{\frac{e_{(23)}t}{1-e_{(13)}t}}{\frac{1}{1-e_{(1)}q^{-1}t}}\quad l=3\\
&\row{A^{\hdelta}_{(14)}(t)}{A^{\hdelta}_{(23)}(t)} =\row{A^{\hdelta}_{(15)}(t)}{A^{\hdelta}_{(24)}(t)} 
\mat{\frac{e_{(15)}t}{1-e_{(14)}t}}{0}{\frac{1}{1-e_{(14)}t}}{\frac{1}{1-e_{(23)}t}}\quad l=4\\
&\row{A^{\hdelta}_{(15)}(t)}{A^{\hdelta}_{(24)}(t)} =\row{A^{\hdelta}_{(25)}(t)}{A^{\hdelta}_{(34)}(t)} 
\mat{\frac{1}{1-e_{(15)}t}}{\frac{e_{(25)}t}{1-e_{(24)}t}}{0}{\frac{1}{1-e_{(24)}t}}\quad l=5\\
&\row{A^{\hdelta}_{(25)}(t)}{A^{\hdelta}_{(34)}(t)} =\row{A^{\hdelta}_{(35)}(t)}{A^{\hdelta}_{(5)}(t)} 
\mat{\frac{1}{1-e_{(25)}t}}{\frac{1}{1-e_{(34)}t}}{0}{\frac{e_{(5)}t}{1-e_{(34)}t}}\quad l=6\\
&\row{A^{\hdelta}_{(35)}(t)}{A^{\hdelta}_{(5)}(t)} =\row{A^{\hdelta}_{(45)}(t)}{A^{\hdelta}_{(4)}(t)} 
\mat{\frac{1}{1-e_{(35)}t}}{0}{\frac{e_{(4)}t}{1-e_{(35)}t}}{\frac{1}{1-e_{(5)}t}}\quad l=7\\
&\row{A^{\hdelta}_{(45)}(t)}{A^{\hdelta}_{(4)}(t)} =\row{A^{\hdelta}_{(0)^1}(t)}{A^{\hdelta}_{(3)}(t)} 
\mat{\frac{e_{(0)}qt}{1-e_{(45)}t}}{0}{\frac{1}{1-e_{(45)}t}}{\frac{1}{1-e_{(4)}t}} \quad l=8
\end{split}
\end{equation}

\subsection{Special cases}\label{S:specialcases}
The formulas for characters $A_{\hdelta}^{\hdelta'}(z,q,t)$ undergo significant simplification after specialization $z=1,q=1$. The best result can be obtained when we further restrict the range of $\hdelta'$. We shall explore this presently. 

The set of weights $J$ (\ref{E:J}) 
is of interest because 
it is possible to write a vary simple recursive relation between corresponding Poincar\'{e} series.

\begin{proposition}\label{E:recursionDelgen}
The following recursions hold:
\[\begin{split}
&A_{\hgamma}^{{(5)^{r}}}(t) ={\frac { \left( 1-e_{{15}}{q}^{r}t-e_{{14}}e_{{23}}{q}^{2\,r}{t}^{2}
+e_{{15}}e_{{14}}e_{{23}}{q}^{3\,r}{t}^{3} \right) }{
 \left( 1-e_{{13}}{q}^{r}t \right)  \left( 1-e_{{34}}{q}^{r}t
 \right)  \left( 1-e_{{5}}{q}^{r}t \right)  \left( 1-e_{{23}}{q}^{r}
t \right) }}A_{\hgamma}^{{(15)^{r}}}(t) +\\
&{\frac {e_{{1}}{q}^{r-1}t}{ \left( 1-e_{{13}}
{q}^{r}t \right)  \left( 1-e_{{34}}{q}^{r}t \right)  \left( 1-e_{{5}
}{q}^{r}t \right)  \left( 1-e_{{23}}{q}^{r}t \right) }}A_{\hgamma}^{{(1)^{r-1}}}(t)\quad \hgamma\leq (1)^{r-1}
\end{split}\]

\[\begin{split}&A_{\hgamma}^{{(15)^{r}}}(t) ={\frac { \left( 1-e_{{1}}{q}^{r-1}t-e_{{2}}e_{{12}}{q}^{2\,r-1}{t}^{2}
+e_{{1}}e_{{2}}e_{{12}}{q}^{3\,r-2}{t}^{3} \right) }{
 \left( 1-e_{{13}}{q}^{r}t \right)  \left( 1-e_{{14}}{q}^{r}t
 \right)  \left( 1-e_{{15}}{q}^{r}t \right)  \left( 1-e_{{12}}{q}^{r
}t \right) }}A_{\hgamma}^{{(1)^{r-1}}}(t)+\\
&+{\frac {e_{{0}}{q}^{r}t}{ \left( 1-e_{{13}}{q}
^{r}t \right)  \left( 1-e_{{14}}{q}^{r}t \right)  \left( 1-e_{{15}}{
q}^{r}t \right)  \left( 1-e_{{12}}{q}^{r}t \right) }}A_{\hgamma}^{(0)^r}(t) \quad \hgamma\leq (0)^{r}
\end{split}\]

\[\begin{split}
&A_{\hgamma}^{{(1)^{r}}}(t) ={\frac { \left( 1-e_{{0}}{q}^{r+1}t-e_{{45}}e_{{4}}{q}^{2\,r}{t}^{2}+e
_{{0}}e_{{45}}e_{{4}}{q}^{3\,r+1}{t}^{3} \right) }{ \left( 
1-e_{{3}}{q}^{r}t \right)  \left( 1-e_{{2}}{q}^{r}t \right)  \left( 
1-e_{{1}}{q}^{r}t \right)  \left( 1-e_{{4}}{q}^{r}t \right) }}A_{\hgamma}^{{(0)^{r+1}}}(t)+\\
&+{\frac {e_{{5}}{q}^{r}t}{ \left( 1-e_{{3}}{q}^{r}t \right) 
 \left( 1-e_{{2}}{q}^{r}t \right)  \left( 1-e_{{1}}{q}^{r}t \right) 
 \left( 1-e_{{4}}{q}^{r}t \right) }}A_{\hgamma}^{{(5)^{r}}}(t)\quad \hgamma\leq (5)^{r}
\end{split}\]

\[\begin{split}
&A_{\hgamma}^{{(0)^{r}}}(t)={\frac { \left( {q}^{3}-{q}^{2+r}e_{{5}}t-{q}^{2\,r+1}e_{{34}}{t}^{2
}e_{{25}}+{q}^{3\,r}e_{{34}}{t}^{3}e_{{5}}e_{{25}} \right) 
}{ \left( q-e_{{35}}{q}^{r}t \right)  \left( q-e_{{45}}{q}^{r}t
 \right)  \left( 1-e_{{0}}{q}^{r}t \right)  \left( q-e_{{25}}{q}^{r}
t \right) }}A_{\hgamma}^{{(5)^{r-1}}}(t)+\\
&+{\frac {{q}^{2+r}e_{{15}}t}{ \left( q-e_{{35
}}{q}^{r}t \right)  \left( q-e_{{45}}{q}^{r}t \right)  \left( 1-e_{{0
}}{q}^{r}t \right)  \left( q-e_{{25}}{q}^{r}t \right) }}A_{\hgamma}^{{(15)^{r-1}}}(t) \quad \hgamma\leq (15)^{r-1}
\end{split}
\]

\end{proposition}

\begin{proof}
We give details on the proof in case of $A_{\hdelta}^{{(5)^{r}}}(t)$ only. All other cases can be treated similarly so the proofs are omitted.
Iteration of the formula  (\ref{E:tail}) leads to 
\begin{equation}\label{E:firsteq}
A_{\hdelta}^{{(5)^{r}}}(t)={\frac {1}{ \left( 1-e_{{15}}{q}^{r}t \right)  \left( 1-e_{{
14}}{q}^{r}t \right) }}A_{\hdelta}^{{(13)^{r}}}(t)
\end{equation}

Formulas 
\[\begin{split}&A_{\hdelta}^{{(13)^{r}}}(t)={\frac {1}{1-e_{{13}}{q}^{r}t}}A_{\hdelta}^{{(12)^{r}}}(t)+{\frac {e_{{2}}{q}
^{r-1}t}{1-e_{{13}}{q}^{r}t}}A_{\hdelta}^{{(2)^{r-1}}}(t)\\
&A_{\hdelta}^{{(12)^{r}}}(t)={\frac {e_{{0}}{q}^{r}t}{1-e_{{12}}{q}^{r}t}}A_{\hdelta}^{{(0)^{r}}}(t)+{
\frac {1}{1-e_{{12}}{q}^{r}t}}A_{\hdelta}^{{(2)^{r-1}}}(t)
\end{split}\]

are the special cases of (\ref{E:reduction}).
Expressions
\[A_{\hdelta}^{{(2)^{r-1}}}(t)= \left( 1-e_{{1}}{q}^{r-1}t \right)  \left( 1-e_{{2}}{q}^{
r-1}t \right) A_{\hdelta}^{{(1)^{r-1}}}(t)
\]
\[A_{\hdelta}^{{(2)^{r-1}}}(t)= \left( 1-e_{{1}}{q}^{r-1}t \right) A_{\hdelta}^{{(1)^{r-1}}}(t)\]
are iterations of (\ref{E:tail}) applied backwards. We obtain our result by making substitutions of the presented identities in the order they are written into (\ref{E:firsteq}).
\end{proof}
Let us introduce notations: $B_{\hgamma}^r(t)=A_{\hgamma}^{\hdelta_r}(1,1,t)$, $\hdelta_r\in J$.
\begin{corollary}
Under above assumptions  on labeling weights  the following recursions hold:
\[B_{\hgamma}^r(t)
=\frac{1+t}{(1-t)^2}B_{\hgamma}^{r-1}(t)+\frac{t}{(1-t)^4}B_{\hgamma}^{r-2}(t) \quad \delta_{r-2}\geq \hgamma   \]
\end{corollary}
\begin{proof}
This recursion is a specialization of formulas from Proposition \ref{E:recursionDelgen}.
\end{proof}

It follows from results of Examples \ref{E:proj}  and \ref{E:quad}  that \[B_{(0)}^{0}(t)=A_{(0)}^{(15)}(1,1,t)=\frac{1}{(1-t)^5}\] and  \[B_{(0)}^{5}(t)=A_{(0)}^{(5)}(1,1,t)=\frac{1-t^2}{(1-t)^8}\]

\subsection{Delannoy polynomials}\label{S:Delannoyp}
Delannoy polynomials (cf. \cite{WDelannoy}) $D_n(t)=\sum_{k=0}^nD_{k,n-k}t^k, \deg D_n=n$ is a  recursive defined system
\[\begin{split}
&D_n(t)=(1+t)D_{n-1}(t)+tD_{n-2}(t)\\
&D_0=1,\quad D_1(t)=1+t\end{split}\]
Here is a first few polynomials:
\[
\begin{split}
&D_0(t)=1\\
&D_3(t)=1+t\\
&D_2(t)=1+3\,t+{t}^{2}\\
&D_3(t)=1+5\,t+5\,{t}^{2}+{t}^{3}\\
&D_4(t)=1+7\,t+13\,{t}^{2}+7\,{t}^{3}+{t}^{4}\\
&D_5(t)=1+9\,t+25\,{t}^{2}+25\,{t}^{3}+9\,{t}^{4}+{t}^{5}\\
&D_6(t)=1+11\,t+41\,{t}^{2}+63\,{t}^{3}+41\,{t}^{4}+11\,{t}^{5}+{t}^{6}\\
&D_7(t)=1+13\,t+61\,{t}^{2}+129\,{t}^{3}+129\,{t}^{4}+61\,{t}^{5}+13\,{t}^{6}+
{t}^{7}\\
&D_8(t)=1+15\,t+85\,{t}^{2}+231\,{t}^{3}+321\,{t}^{4}+231\,{t}^{5}+85\,{t}^{6}
+15\,{t}^{7}+{t}^{8}\\
&D_9(t)=1+17\,t+113\,{t}^{2}+377\,{t}^{3}+681\,{t}^{4}+681\,{t}^{5}+377\,{t}^{
6}+113\,{t}^{7}+17\,{t}^{8}+{t}^{9}\\
&D_{10}(t)=1+19\,t+145\,{t}^{2}+575\,{t}^{3}+1289\,{t}^{4}+1683\,{t}^{5}+1289\,{t
}^{6}+575\,{t}^{7}+145\,{t}^{8}+19\,{t}^{9}+{t}^{10}\\
&D_{11}(t)=1+21\,t+181\,{t}^{2}+833\,{t}^{3}+2241\,{t}^{4}+3653\,{t}^{5}+3653\,{t
}^{6}+2241\,{t}^{7}+833\,{t}^{8}+181\,{t}^{9}+21\,{t}^{10}+{t}^{11}\\
&\cdots
\end{split}\]
Delannoy numbers form a triangle $\{D_{n,m}\}$.
A number $D_{n,m}$ describes the number of paths from the southwest corner $(0, 0)$ of a rectangular grid to the northeast corner $(m, n)$, using only single steps north, northeast, or east. 

Let \[B_r(t)=\frac{D_r(t)}{(1-t)^{5+2r}}\] The functions $B_r$ satisfy recurrence (\ref{E:Delan})
\begin{corollary}
Numerator of $B_{(0)}^r(t)$ is a Delannoy polynomial $D_r(t), r\geq 0$.
\end{corollary}
\begin{proof}
Sequences $B_{(0)}^r(t)$ and $B_r(t)$ satisfy the same recurrence relation and initial conditions.
\end{proof}

The formula for generating function
\[\sum_{r\geq 0} B_r(t)s^r={\frac {1}{ \left( -1+t \right)  \left( s-st-s{t}^{2}+s{t}^{3}+t{s}^{2
}-1+4\,t-6\,{t}^{2}+4\,{t}^{3}-{t}^{4} \right) }}\]
follows easily form the equation (\ref{E:Delan}). Its verification is left as an exercise. 

\bigskip
 \bigskip
 
 {\bf \Huge Appendix}
 \appendix
 \section{Some automorphisms of $\hat{A}$}\label{S:Aut}
  \begin{remark}
Recall that $\mathrm{Pin}(10)$-group, the universal cover of $\mathrm{O}(10)$, is generated by elements $e\in W+W'$ in Clifford algebra $Cl(W+W')$, that satisfy 
\begin{equation}\label{E:square}
e^2=1.
\end{equation}
 In the basis $v_1,\dots,v_5$ of $W$ and the dual basis $v_{1^*},\dots,v_{5^*}$ in $W^*$  element $e_{i}(a)=av_i+a^{-1}v_{i^*}$ satisfies (\ref{E:square}).  in addition $e_{i}(a)e_{j}(b)=-e_{j}(b)e_{i}(a)$. An element $e_{i}(a)e_{j}(b)$ belongs to $\Spin(10)\subset \GL(S)$.
 \end{remark}
The group $\Spin_{\mathbb{C}}(10)=\Spin(10)\times \mathbb{C}^{\times}/(-1,-1)$ has a tautological spinorial representation. 
Its Lie algebra acts in $S$ by operators $v_iv_j,v_iv_{j^*},v_{i^*}v_{j^*},1$ (see e.g. \cite{ChevalleySpinors},p131).
 
  Let  $u\in \Spin_{\mathbb{C}}(10)$ to be $u=e_2e_3e_4e_5$. Equation \ref{E:square} and anti-commutativity of $e_i$ imply $u^2=1$.
 \begin{lemma}\label{L:gaction}
 We have the following table of $s$-action on basis elements:
 \begin{equation}
 \begin{split}
 &u(\theta_{(0)})=\theta_{(1)};u(\theta_{(12)})=-\theta_{(2)};u(\theta_{(13)})=\theta_{(3)};u(\theta_{(14)})=-\theta_{(4)}\\
 &u(\theta_{(15)})=\theta_{(5)};u(\theta_{(23)})=-\theta_{(45)};u(\theta_{(24)})=\theta_{(35)};u(\theta_{(25)})=-\theta_{(34)}.
 \end{split}
 \end{equation} 
 We give only a half of the formulas. The remaining half  can be recovered from the condition $u^2=1$.
 \end{lemma}
 \begin{proof}
 Direct check.
 \end{proof}
 
 We see that $u$ permutes weight spaces and induces a transformation of the set $\rE$, which we by abuse of notations denote by $u$.
\begin{definition}
A invertible transformation $h:M\rightarrow M$ of a poset $M$ is an anti-automorphism if $\alpha\leq \beta$ implies $h(\alpha)\geq h(\beta)$.
\end{definition}
 \begin{proposition}\label{P:centralsym}
 The poset $E$ has a anti-automorphism $u$. It is realized as a  symmetry of the graph $\rE$ (see Picture \ref{P:kxccvgw1}) with respect to its geometric center. This symmetry can be realized by a linear transformation $u\in \Spin_{\mathbb{C}}(10)$, that satisfies $u\theta_{\alpha}=\pm\theta_{u(\alpha)},\alpha\in \rE$.
  The operator $u\in \Spin(10)$  defines an automorphism of algebra $A$.
 \end{proposition}
 \begin{proof}
 The first assertion follows from Lemma \ref{L:gaction}.
 The operator $u$ belongs to $\Spin_{\mathbb{C}}(10)$. The operator  must preserve linear space spanned by relations (\ref{E:equations}) in $P$ and defines an automorphism of algebra $A$.
 \end{proof}
 
The automorphism $u$ acts on $\hat{P}$ and preserves ideal of relations $I$. We denote by $\hat{u}$ the composition of $u$ with transformation $\lambda^{\alpha}(z)\rightarrow \lambda^{\alpha}(\frac{1}{z})$. The automorphism $\hat{u}$ permutes weight spaces and defines a transformation of  $\hat{\rE}$, which we also denote by $\hat{u}$.
\begin{proposition}
Transformation $\hat{u}$ defines an anti-automorphism  $\hat{\rE}$ 
\begin{equation}\label{E:inv}
\hat{u}:(\alpha)^l\rightarrow u(\alpha)^{-l}. 
\end{equation}
It comes from an automorphism $\hat{u}$ of the pair $I\subset \hat{P}$.
\end{proposition}
\begin{proof}
Note that $\hat{u}\rE_r=\rE_{-r}$. By Lemma \ref{E:dom} and by 
 virtue of  identity $\hat{u}T\hat{u}=T^{-1}$, it suffice to verify axioms of an  anti-automorphism for $\hat\alpha,\hat\beta\in [(0),(1)^1]$. We leave this  as an exercise.
\end{proof}

The shift operator $T$ is continuous  on $\hat{\rE}$ and   defines an automorphism of  $\hat P$:
\begin{equation}\label{E:shift}
\lambda_{\alpha}(z)\rightarrow \frac{\lambda_{\alpha}(z)}{z}.
\end{equation}  On quadrics  $\Gamma^{\s^k}\in \hat P$    (\ref{E:relgenfunction})with no restrictions on $l,l'$ indices  it acts by the formula $T\Gamma^{\s^k}=\Gamma^{\s^{k+2}}$. Thus $T$  is an automorphism of the pair $I\subset \hat P$. 

\section{Lists of obstructions}\label{A:obstr}
Below is a list of pairs of noncommutative obstructions in $\rE$ (see Definition \ref{D:obstruction})
\begin{equation}\label{E:obstructionsord}
\begin{split}
&o_{(0)}=\{((4),((5)(45))), ((5),((4)(45)))\}\\
&o_{(12)}=\{ ((45),((5)(35))), ((35),((5)(45)))\}\\
&o_{(13)}=\{ ((45),((5)(25))), ((25),((5)(45)))\}\\
&o_{(23)}=\{ ((45),((5)(15))), ((15),((5)(45)))\}\\
&o_{(14)}=\{ ((35),((5)(25))), ((25),((5)(35)))\}\\
&o_{(24)}=\{ ((35),((5)(15))), ((15),((5)(35)))\}\\
&o_{(15)}=\{ ((34),((5)(25))), ((5),((5)(35)))\}\\
&o_{(34)}=\{ ((25),((5)(15))), ((15),((5)(25)))\}\\
&o_{(25)}=\{ ((34),((5)(15))), ((5),((15)(34)))\}\\
&o_{(5)}=\{ ((15),((25)(34))), ((25),((15)(34)))\}\\
&o_{(35)}=\{ ((24),((5)(15))), ((5),((15)(24)))\}\\
&o_{(4)}=\{ ((24),((15)(34))), ((34),((15)(24)))\}\\
&o_{(45)}=\{ ((23),((5)(15))), ((5),((15)(23)))\}\\
&o_{(3)}=\{ ((23),((15)(24))), ((34),((15)(23)))\}\\
&o_{(2)}=\{ ((23),((15)(24))), ((24),((15)(23)))\}\\
&o_{(1)}=\{ ((14),((15)(23))), ((15),((14)(23)))\}
\end{split}
\end{equation}
The sets $o_{\alpha}$ carry a label $\alpha$, a weight in a $\so_{10}$-spinor representation. Expressions $h_{\alpha}$ (\ref{E:fierzksimple}) carry the same label. To make connection of $o_{\alpha}$ with $h_{\alpha}$ more explicitly we identify the later with an element in $P\otimes I$ by means of substitution $x_{\s}\rightarrow \Gamma^{\s}$.   Then $o_{\alpha}$ appears as a collection  indices of noncommutative obstructions in $h_{\alpha}$.

A list of noncommutative obstruction in $[(0),(1)^1]\subset \hat{\rE}$, that are  not in the preceding list and not in $[(0)^1,(1)^1]$ contains elements of the form $(\alpha^i,(\beta^j,\gamma^k))$. Indices $i,j,k\geq 0$ satisfy $i+j+k=l=1$ or $2$. The table below describes the $l=1$ case
\begin{equation}\label{E:ncobstraf}
\begin{split}
&o_{(1)^1}=\{((15),((5),(0)^1)),((0)^1,((15),(5)))\}\\
&o_{(2)^1}=\{((25),((5),(0)^1)),((0)^1,((25),(5)))\}\\
&o_{(3)^1}=\{((35),((5),(0)^1)),((0)^1,((35),(5)))\}\\
&o_{(4)^1}=\{((45),((5),(0)^1)),((0)^1,((45),(5)))\}\\
&o_{(4)^1}=\{((45),((4),(0)^1)),((0)^1,((45),(4)))\}\\
&o_{(45)^1}=\{((5),((0)^1,(4))),((4),((0)^1,(5)))\}\\
&o_{(35)^1}=\{((5),((0)^1,(3))),((3),((0)^1,(5)))\}\\
&o_{(25)^1}=\{((5),((0)^1,(2))),((2),((0)^1,(5)))\}\\
&o_{(15)^1}=\{((5),((0)^1,(1))),((1),((0)^1,(5)))\}\\
&o_{(34)^1}=\{((4),((0)^1,(3))),((3),((0)^1,(4)))\}\\
&o_{(24)^1}=\{((4),((0)^1,(2))),((2),((0)^1,(4)))\}\\
&o_{(14)^1}=\{((4),((0)^1,(1))),((1),((0)^1,(4)))\}\\
&o_{(23)^1}=\{((3),((0)^1,(2))),((2),((0)^1,(3)))\}\\
&o_{(13)^1}=\{((3),((0)^1,(1))),((1),((0)^1,(3)))\}\\
&o_{(12)^1}=\{((2),((0)^1,(1))),((1),((0)^1,(2)))\}\\
&o_{(0)^1}=\{((2),((12)^1,(1))),((1),((12)^1,(3)))\}\\
\end{split}
\end{equation}
Obstructions with $l=2$ can be transformed by symmetries $T$ and $\hat u$ into obstructions with $l=1$. Because of this, the above table is sufficient for our purposes.

\end{document}